\documentclass{tlp}

\usepackage{amsmath}
\usepackage{latexsym}
\usepackage{times}
\usepackage{color}
\usepackage{multirow}

\newtheorem{definition}{Definition}[section]
\newtheorem{example}{Example}[section]
\newtheorem{lemma}{Lemma}[section]
\newtheorem{theorem}[lemma]{Theorem}
\newtheorem{corollary}[lemma]{Corollary}
\newtheorem{proposition}[lemma]{Proposition}
\newtheorem{algorithm}{Algorithm}[section]

\begin{document}
\title[Termination Prediction for General Logic Programs]
{Termination Prediction for General Logic Programs}
\author[Y.-D. Shen, D.De Schreye, D.Voets]
{YI-DONG SHEN \\
State Key Laboratory of Computer Science, Institute of Software,
Chinese Academy of Sciences, Beijing 100190, China\\
\email{ydshen@ios.ac.cn}
\and DANNY DE SCHREYE, DEAN VOETS\\
Department of Computer Science, Celestijnenlaan 200 A, B-3001 Heverlee, Belgium\\ 
\email{\{danny.deschreye, dean.voets\}@cs.kuleuven.ac.be}
}

\submitted{10 October 2007}
 \revised{11 September 2008, 3 April 2009}
 \accepted{12 May 2009}

\maketitle

\begin{abstract}
We present a heuristic framework
for attacking the undecidable termination problem of
logic programs, as an alternative to current termination/non-termination
proof approaches. We introduce an idea of termination prediction,
which predicts termination of a logic program in case that
neither a termination nor a non-termination proof is applicable. 
We establish a necessary and sufficient characterization of infinite
(generalized) SLDNF-derivations with arbitrary (concrete or moded) queries, and 
develop an algorithm that predicts termination of 
general logic programs with arbitrary non-floundering queries.
We have implemented a termination prediction tool and obtained 
quite satisfactory experimental results. Except for five programs
which break the experiment time limit, our prediction is $100\%$
correct for all 296 benchmark programs of the Termination Competition 2007,
of which eighteen programs cannot be proved by any of the existing
state-of-the-art analyzers like
AProVE07, NTI, Polytool and TALP. 
\end{abstract} 

\section{Introduction} 
\label{intr}
Termination is a fundamental problem in logic programming 
with SLDNF-resolution as the query evaluation mechanism \cite{clark78,Ld87}, 
which has been extensively studied
in the literature (see, e.g., \citeN{DD93} for a survey and some recent papers 
\cite{Apt93,BCGGV05,DDV99,GC05,LS97,MN01,PM06,SGST06}). 
Since the termination problem is undecidable, existing algorithms/tools
either focus on computing sufficient termination conditions which once 
satisfied, lead to a positive conclusion {\em terminating} 
\cite{AZ96,BCRE02,DLSS01,GC05,GST06,LSS97,Mar96-1,MB05,MN01,OCM00,SGST06},
or on computing sufficient non-termination conditions
which lead to a negative conclusion {\em non-terminating} \cite{PM06,Pay06}. 
For convenience, we call the former computation a 
{\em termination proof}, and the latter a {\em non-termination proof}.
Due to the nature of undecidability, there must be situations in which neither
a termination proof nor a non-termination proof can apply;
i.e., no sufficient termination/non-termination conditions 
are satisfied so that the user would get no conclusion
(see the results of the Termination Competition 2007 which is available at
http://www.lri.fr/\verb+~+marche/termination-competition/2007).
We observe that in such a situation, it is particularly useful
to compute a heuristic conclusion indicating likely 
termination or likely non-termination, which guides the user
to continue to improve his program towards termination.
To the best of our knowledge, however,
there is no existing heuristic approach available.  
The goal of the current paper is then to develop such a heuristic framework.

We propose an idea of {\em termination prediction}, as depicted 
in Figure \ref{framework}. In the case that neither a termination
nor a non-termination proof is applicable, we appeal to a heuristic algorithm 
to predict possible termination or non-termination. 
The prediction applies to general logic programs with 
concrete or moded queries.

\begin{figure}[htb]
\begin{center}
\setlength{\unitlength}{3947sp}%
\begingroup\makeatletter\ifx\SetFigFont\undefined%
\gdef\SetFigFont#1#2#3#4#5{%
  \reset@font\fontsize{#1}{#2pt}%
  \fontfamily{#3}\fontseries{#4}\fontshape{#5}%
  \selectfont}%
\fi\endgroup%
\begin{picture}(5901,4728)(979,-4477)
\thinlines
{\color[rgb]{0,0,0}\put(991,-3436){\framebox(2610,900){}}
}%
{\color[rgb]{0,0,0}\put(5071,-1720){\oval(240,240)[bl]}
\put(5071,-1456){\oval(240,240)[tl]}
\put(6748,-1720){\oval(240,240)[br]}
\put(6748,-1456){\oval(240,240)[tr]}
\put(5071,-1840){\line( 1, 0){1677}}
\put(5071,-1336){\line( 1, 0){1677}}
\put(4951,-1720){\line( 0, 1){264}}
\put(6868,-1720){\line( 0, 1){264}}
}%
{\color[rgb]{0,0,0}\put(1471,-4345){\oval(240,240)[bl]}
\put(1471,-4081){\oval(240,240)[tl]}
\put(3148,-4345){\oval(240,240)[br]}
\put(3148,-4081){\oval(240,240)[tr]}
\put(1471,-4465){\line( 1, 0){1677}}
\put(1471,-3961){\line( 1, 0){1677}}
\put(1351,-4345){\line( 0, 1){264}}
\put(3268,-4345){\line( 0, 1){264}}
}%
{\color[rgb]{0,0,0}\put(5071,-370){\oval(240,240)[bl]}
\put(5071,-106){\oval(240,240)[tl]}
\put(6748,-370){\oval(240,240)[br]}
\put(6748,-106){\oval(240,240)[tr]}
\put(5071,-490){\line( 1, 0){1677}}
\put(5071, 14){\line( 1, 0){1677}}
\put(4951,-370){\line( 0, 1){264}}
\put(6868,-370){\line( 0, 1){264}}
}%
{\color[rgb]{0,0,0}\put(991,-661){\framebox(2610,900){}}
}%
{\color[rgb]{0,0,0}\put(991,-2011){\framebox(2610,900){}}
}%
{\color[rgb]{0,0,0}\put(5071,-3145){\oval(240,240)[bl]}
\put(5071,-2881){\oval(240,240)[tl]}
\put(6748,-3145){\oval(240,240)[br]}
\put(6748,-2881){\oval(240,240)[tr]}
\put(5071,-3265){\line( 1, 0){1677}}
\put(5071,-2761){\line( 1, 0){1677}}
\put(4951,-3145){\line( 0, 1){264}}
\put(6868,-3145){\line( 0, 1){264}}
}%
{\color[rgb]{0,0,0}\put(3601,-1561){\vector( 1, 0){1284}}
}%
{\color[rgb]{0,0,0}\put(2251,-2011){\vector( 0,-1){444}}
}%
{\color[rgb]{0,0,0}\put(3601,-2986){\vector( 1, 0){1284}}
}%
{\color[rgb]{0,0,0}\put(2251,-3436){\vector( 0,-1){444}}
}%
{\color[rgb]{0,0,0}\put(3601,-211){\vector( 1, 0){1284}}
}%
{\color[rgb]{0,0,0}\put(2251,-661){\vector( 0,-1){444}}
}%
\put(3976,-1486){\makebox(0,0)[lb]{\smash{{\SetFigFont{8}{9.6}{\rmdefault}{\mddefault}{\updefault}{\color[rgb]{0,0,0}Yes}%
}}}}
\put(2326,-2236){\makebox(0,0)[lb]{\smash{{\SetFigFont{8}{9.6}{\rmdefault}{\mddefault}{\updefault}{\color[rgb]{0,0,0}No}%
}}}}
\put(1426,-1486){\makebox(0,0)[lb]{\smash{{\SetFigFont{10}{12.0}{\rmdefault}{\mddefault}{\updefault}{\color[rgb]{0,0,0}{\bf Non-Termination Proof}}%
}}}}
\put(3976,-2911){\makebox(0,0)[lb]{\smash{{\SetFigFont{8}{9.6}{\rmdefault}{\mddefault}{\updefault}{\color[rgb]{0,0,0}Yes}%
}}}}
\put(2326,-3661){\makebox(0,0)[lb]{\smash{{\SetFigFont{8}{9.6}{\rmdefault}{\mddefault}{\updefault}{\color[rgb]{0,0,0}No}%
}}}}
\put(1501,-2911){\makebox(0,0)[lb]{\smash{{\SetFigFont{10}{12.0}{\rmdefault}{\mddefault}{\updefault}{\color[rgb]{0,0,0}{\bf Termination Prediction}}%
}}}}
\put(3976,-136){\makebox(0,0)[lb]{\smash{{\SetFigFont{8}{9.6}{\rmdefault}{\mddefault}{\updefault}{\color[rgb]{0,0,0}Yes}%
}}}}
\put(2326,-886){\makebox(0,0)[lb]{\smash{{\SetFigFont{8}{9.6}{\rmdefault}{\mddefault}{\updefault}{\color[rgb]{0,0,0}No}%
}}}}
\put(1501,-136){\makebox(0,0)[lb]{\smash{{\SetFigFont{10}{12.0}{\rmdefault}{\mddefault}{\updefault}{\color[rgb]{0,0,0}{\bf Termination Proof}}%
}}}}
\put(1576,-4261){\makebox(0,0)[lb]{\smash{{\SetFigFont{9}{10.8}{\rmdefault}{\mddefault}{\updefault}{\color[rgb]{0,0,0}$predicted$-$terminating$}%
}}}}
\put(1201,-436){\makebox(0,0)[lb]{\smash{{\SetFigFont{8}{9.6}{\rmdefault}{\mddefault}{\updefault}{\color[rgb]{0,0,0}Sufficient termination conditions satisfied?}%
}}}}
\put(1051,-1786){\makebox(0,0)[lb]{\smash{{\SetFigFont{8}{9.6}{\rmdefault}{\mddefault}{\updefault}{\color[rgb]{0,0,0}Sufficient non-termination conditions satisfied?}%
}}}}
\put(1126,-3211){\makebox(0,0)[lb]{\smash{{\SetFigFont{8}{9.6}{\rmdefault}{\mddefault}{\updefault}{\color[rgb]{0,0,0}Potential infinite SLDNF-derivations found?}%
}}}}
\put(5026,-3061){\makebox(0,0)[lb]{\smash{{\SetFigFont{9}{10.8}{\rmdefault}{\mddefault}{\updefault}{\color[rgb]{0,0,0}$predicted$-$non$-$terminating$}%
}}}}
\put(5401,-1636){\makebox(0,0)[lb]{\smash{{\SetFigFont{9}{10.8}{\rmdefault}{\mddefault}{\updefault}{\color[rgb]{0,0,0}$non$-$terminating$}%
}}}}
\put(5551,-286){\makebox(0,0)[lb]{\smash{{\SetFigFont{9}{10.8}{\rmdefault}{\mddefault}{\updefault}{\color[rgb]{0,0,0}$terminating$}%
}}}}
\end{picture}%
\end{center} 
\caption{A framework for handling the termination problem.}
\label{framework} 
\end{figure}
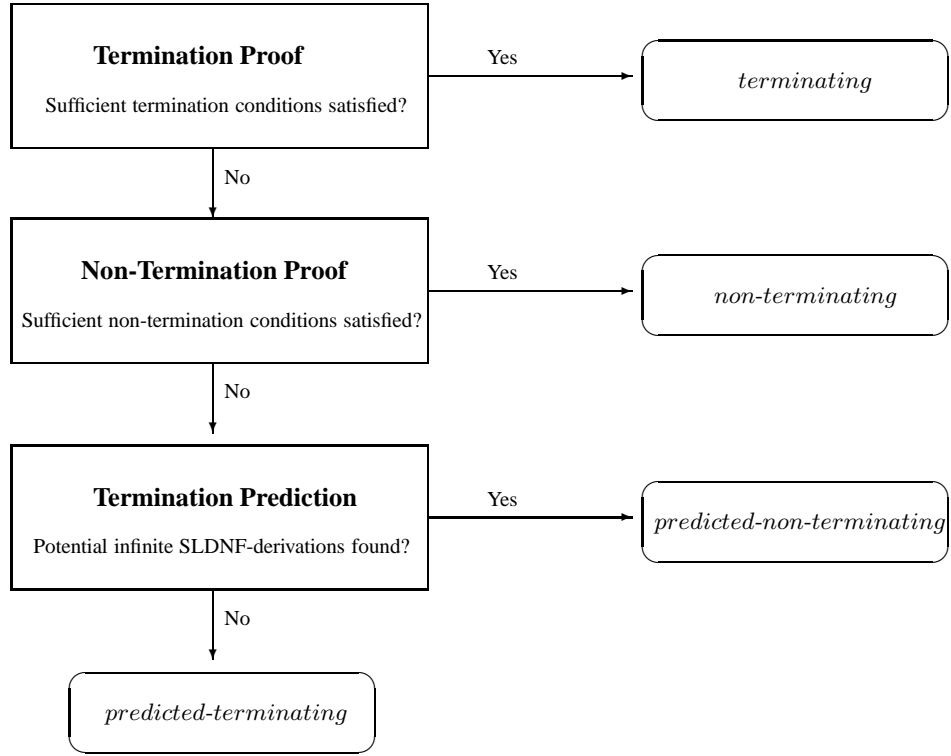 

We develop a framework for predicting termination of general logic
programs with arbitrary (i.e., concrete or moded) queries. 
The basic idea is that we establish a characterization of 
infinite (generalized) SLDNF-derivations with arbitrary queries. 
Then based on the characterization, we design 
a complete loop checking mechanism, 
which cuts all infinite SLDNF-derivations.
Given a logic program and a query, 
we evaluate the query by applying SLDNF-resolution while performing loop checking.
If the query evaluation proceeds without encountering potential infinite derivations,
we predict {\em terminating} for this query; 
otherwise we predict {\em non-terminating}.

The core of our termination prediction is a characterization of infinite
SLDNF-derivations with arbitrary queries. 
In \citeN{shen-tocl}, a characterization
is established for general logic programs with concrete queries. 
This is far from enough for termination prediction; a characterization of infinite
SLDNF-derivations for moded queries is required. 
Moded queries are the most commonly used query form in static termination analysis.
A moded query contains (abstract) atoms like $p({\cal I},T)$
where $T$ is a term (i.e., a constant, variable or function) and ${\cal I}$
is an input mode. An {\em input mode} stands for an arbitrary ground (i.e.\ variable-free)
term, so that to prove that a logic program terminates
for a moded query $p({\cal I},T)$ is to prove that 
the program terminates for any (concrete) query $p(t,T)$
where $t$ is a ground term. 

It is nontrivial to characterize infinite
SLDNF-derivations with moded queries. 
The first challenge we must address is how to formulate an SLDNF-derivation
for a moded query $Q_0$, as the standard SLDNF-resolution
is only for concrete queries \cite{clark78,Ld87}. 
We will introduce a framework called a {\em moded-query forest},
which consists of all (generalized) SLDNF-trees 
rooted at an instance of $Q_0$ (the instance is $Q_0$
with each input mode replaced by a ground term). 
An SLDNF-derivation for $Q_0$ is then defined over
the moded-query forest such that a logic program $P$
terminates for $Q_0$ if and only if the moded-query 
forest contains no infinite SLDNF-derivations.

A moded-query forest may have an infinite number of SLDNF-trees, 
so it is infeasible for us to predict
termination of a logic program by traversing the moded-query forest.
To handle this challenge, we will introduce a novel compact approximation for a 
moded-query forest, called a {\em moded generalized SLDNF-tree}. The key idea is to 
treat an input mode as a special meta-variable in the way that during query evaluation,
it can be substituted by a constant or function, but cannot be substituted
by an ordinary variable. 
As a result, SLDNF-derivations for a moded query can be 
constructed in the same way as the ones for a concrete query.
A characterization of infinite SLDNF-derivations for moded queries is then established
in terms of some key properties of a moded generalized SLDNF-tree.

We have implemented a termination prediction tool and obtained 
quite satisfactory experimental results. Except for five programs
which break the experiment time limit, our prediction is $100\%$
correct for all 296 benchmark programs of the Termination Competition 2007,
of which eighteen programs cannot be proved by any of the existing
state-of-the-art analyzers like AProVE07, NTI, Polytool and TALP. 

The paper is organized as follows. Section \ref{pre} 
reviews some basic concepts including generalized SLDNF-trees.
Sections \ref{sec2} and \ref{sec3} present a characterization of 
infinite SLDNF-derivations for concrete and moded queries,
respectively. Section \ref{sec4} introduces a new loop checking mechanism,
and based on it develops an algorithm that predicts termination of 
general logic programs with arbitrary queries.
The termination prediction method is illustrated with
representative examples including ones borrowed from the Termination Competition 2007. 
Section \ref{implementation} describes the implementation
of our termination prediction algorithm and presents experimental results
over the programs of the Termination Competition 2007.
Section \ref{sec5} mentions related work, and Section \ref{sec6} concludes.

\section{Preliminaries}
\label{pre}

We assume the reader is familiar with standard terminology of 
logic programs, in particular with SLDNF-resolution, as described in \citeN{Ld87}.
Variables begin with a capital letter $X,Y,Z,U,V$ or $I$, and predicate, function 
and constant symbols with a lower case letter. 
A term is a constant, a variable, or a function
of the form $f(T_1, ..., T_m)$ where $f$ is a function symbol 
and each $T_i$ is a term. For simplicity, 
we use $\overline{T}$ to denote a tuple of terms $T_1, ..., T_m$.
An atom is of the form 
$p(\overline{T})$ where $p$ is a predicate symbol. 
Let $A$ be an atom/term. 
The size of $A$, denoted $|A|$, is the 
number of occurrences of function symbols, variables and constants in $A$.
Two atoms are called {\em variants} if they are the same
up to variable renaming. 
A literal is an atom $A$ or the negation $\neg A$ of $A$.
 
A (general) logic program $P$ is a finite set
of clauses of the form $A\leftarrow L_1,..., L_n$,
where $A$ is an atom and each $L_i$ is a literal.
Throughout the paper, we consider only Herbrand models.
The Herbrand universe and Herbrand base of $P$ are denoted by
$HU(P)$ and $HB(P)$, respectively.
 
A goal $G_i$ is a headless clause
$\leftarrow L_1,..., L_n$ where each literal $L_j$ is called a subgoal.
The goal, $G_0 =\leftarrow Q_0$, for a query $Q_0$ is called
a top goal. Without loss of generality, we assume
that $Q_0$ consists only of one atom. 
$Q_0$ is a {\em moded query} if some arguments of $Q_0$ 
are input modes (in this case, $Q_0$ is called an {\em abstract} atom); 
otherwise, it is a {\em concrete query}. 
An input mode always begins with a letter $\cal I$.  

Let $P$ be a logic program and $G_0$ a top goal. 
$G_0$ is evaluated by building a {\em generalized SLDNF-tree} $GT_{G_0}$ 
as defined in \citeN{shen-tocl}, 
in which each node is represented by $N_i:G_i$ where $N_i$ is the name 
of the node and $G_i$ is a goal attached to the node.
We do not reproduce the definition of a generalized SLDNF-tree.
Roughly speaking, $GT_{G_0}$ is the set of standard 
SLDNF-trees for $P\cup \{G_0\}$ augmented  
with an ancestor-descendant relation on their subgoals.
Let $L_i$ and $L_j$ be the selected subgoals 
at two nodes $N_i$ and $N_j$, respectively.
$L_i$ is an {\em ancestor} of $L_j$,
denoted $L_i\prec_{anc} L_j$, if the proof of $L_i$ 
goes via the proof of $L_j$.  
Throughout the paper, we choose to use the best-known 
{\em depth-first, left-most} control strategy, as is used in Prolog,
to select nodes/goals and subgoals (it can be adapted
to any other fixed control strategies).  
So by the {\em selected subgoal} in each node $N_i:\leftarrow L_1,..., L_n$,
we refer to the left-most subgoal $L_1$.

Recall that in SLDNF-resolution, 
let $L_i=\neg A$ be a ground negative subgoal selected at $N_i$, 
then (by the negation-as-failure rule \cite{clark78}) a subsidiary 
child SLDNF-tree $T_{N_{i+1}:\leftarrow A}$ 
rooted at $N_{i+1}:\leftarrow A$ will be built to solve $A$.
In a generalized SLDNF-tree $GT_{G_0}$, 
such parent and child SLDNF-trees are connected from $N_i$ to $N_{i+1}$
via a dotted edge ``$\cdot\cdot\cdot\triangleright$" 
(called a {\em negation arc}), and $A$ at $N_{i+1}$
inherits all ancestors of $L_i$ at $N_i$. 
Therefore, a path of a generalized SLDNF-tree may come  
across several SLDNF-trees through dotted edges.
Any such a path starting at the root node $N_0:G_0$ 
of $GT_{G_0}$ is called a {\em generalized  
SLDNF-derivation}. 

We do not consider {\em floundering} queries; i.e.,
we assume that no non-ground negative subgoals are
selected at any node of a generalized SLDNF-tree (see \citeN{shen-tocl}).

Another feature of a generalized SLDNF-tree $GT_{G_0}$ 
is that each subsidiary child SLDNF-tree $T_{N_{i+1}:\leftarrow A}$ in $GT_{G_0}$ 
terminates (i.e. stops expanding its nodes) at the first success leaf.
The intuition behind this is that it is absolutely 
unnecessary to exhaust the remaining branches  
because they would never generate any new answers for $A$ (since $A$ is ground). 
In fact, Prolog executes the same pruning by using a {\em cut} 
operator to skip the remaining branches 
once the first success leaf is generated 
(e.g. see SICStus Prolog at http://www.sics.se
/sicstus/docs/latest4/pdf/sicstus.pdf).
To illustrate, consider the following logic program and top goal: 
\begin{tabbing} 
$\qquad$ \= $P_0:$ $\quad$ \= $p\leftarrow \neg q$. \`$C_{p_1}$\\ 
\>\> $q$.       \`$C_{q_1}$ \\ 
\>\> $q \leftarrow q$.            \`$C_{q_2}$ \\ 
\>          $G_0:$ \> $\leftarrow p.$ 
\end{tabbing} 
The generalized SLDNF-tree 
$GT_{G_0}$ for $P_0\cup \{G_0\}$ is depicted in 
Figure \ref{fig00}. Note that the subsidiary child SLDNF-tree 
$T_{N_2:\leftarrow q}$ terminates at the first success leaf $N_3$,
leaving $N_4$ not further expanded. As a result, all generalized SLDNF-derivations
in $GT_{G_0}$ are finite.
\begin{figure}[h] 
\centering 
\setlength{\unitlength}{3947sp}%
\begingroup\makeatletter\ifx\SetFigFont\undefined%
\gdef\SetFigFont#1#2#3#4#5{%
  \reset@font\fontsize{#1}{#2pt}%
  \fontfamily{#3}\fontseries{#4}\fontshape{#5}%
  \selectfont}%
\fi\endgroup%
\begin{picture}(2653,1692)(1951,-982)
\thinlines
{\color[rgb]{0,0,0}\multiput(4201, 14)(4.68750,-7.03125){33}{\makebox(1.6667,11.6667){\SetFigFont{5}{6}{\rmdefault}{\mddefault}{\updefault}.}}
\put(4351,-211){\vector( 2,-3){0}}
}%
{\color[rgb]{0,0,0}\put(3601,-211){\vector(-2,-3){0}}
\multiput(3751, 14)(-37.50000,-56.25000){4}{\makebox(1.6667,11.6667){\SetFigFont{5}{6}{\rmdefault}{\mddefault}{\updefault}.}}
}%
{\color[rgb]{0,0,0}\multiput(3226,-511)(-4.68750,-7.03125){33}{\makebox(1.6667,11.6667){\SetFigFont{5}{6}{\rmdefault}{\mddefault}{\updefault}.}}
\put(3076,-736){\vector(-2,-3){0}}
}%
{\color[rgb]{0,0,0}\multiput(3526,-511)(4.68750,-7.03125){33}{\makebox(1.6667,11.6667){\SetFigFont{5}{6}{\rmdefault}{\mddefault}{\updefault}.}}
\put(3676,-736){\vector( 2,-3){0}}
}%
{\color[rgb]{0,0,0}\put(2626,-970){\dashbox{60}(1368,684){}}
}%
{\color[rgb]{0,0,0}\put(3976,539){\vector( 0,-1){300}}
}%
\put(4051,389){\makebox(0,0)[lb]{\smash{{\SetFigFont{8}{9.6}{\rmdefault}{\mddefault}{\updefault}{\color[rgb]{0,0,0}$C_{p_1}$}%
}}}}
\put(3601,614){\makebox(0,0)[lb]{\smash{{\SetFigFont{9}{10.8}{\rmdefault}{\mddefault}{\updefault}{\color[rgb]{0,0,0}$N_0$:  $\leftarrow p$}%
}}}}
\put(3601, 89){\makebox(0,0)[lb]{\smash{{\SetFigFont{9}{10.8}{\rmdefault}{\mddefault}{\updefault}{\color[rgb]{0,0,0}$N_1$:  $\leftarrow \neg q$ }%
}}}}
\put(3151,-436){\makebox(0,0)[lb]{\smash{{\SetFigFont{9}{10.8}{\rmdefault}{\mddefault}{\updefault}{\color[rgb]{0,0,0}$N_2$:  $\leftarrow q$}%
}}}}
\put(2851,-586){\makebox(0,0)[lb]{\smash{{\SetFigFont{8}{9.6}{\rmdefault}{\mddefault}{\updefault}{\color[rgb]{0,0,0}$C_{q_1}$}%
}}}}
\put(3601,-586){\makebox(0,0)[lb]{\smash{{\SetFigFont{8}{9.6}{\rmdefault}{\mddefault}{\updefault}{\color[rgb]{0,0,0}$C_{q_2}$}%
}}}}
\put(2701,-886){\makebox(0,0)[lb]{\smash{{\SetFigFont{9}{10.8}{\rmdefault}{\mddefault}{\updefault}{\color[rgb]{0,0,0}$N_3$:  $\Box_t$}%
}}}}
\put(3376,-886){\makebox(0,0)[lb]{\smash{{\SetFigFont{9}{10.8}{\rmdefault}{\mddefault}{\updefault}{\color[rgb]{0,0,0}$N_4$:  $\leftarrow q$}%
}}}}
\put(4276,-436){\makebox(0,0)[lb]{\smash{{\SetFigFont{9}{10.8}{\rmdefault}{\mddefault}{\updefault}{\color[rgb]{0,0,0}$N_5:\Box_f$}%
}}}}
\put(1951,-736){\makebox(0,0)[lb]{\smash{{\SetFigFont{10}{12.0}{\rmdefault}{\mddefault}{\updefault}{\color[rgb]{0,0,0}$T_{N_2: \leftarrow q}$}%
}}}}
\end{picture}%
\caption{The generalized SLDNF-tree $GT_{G_0}$ of $P_0$.}\label{fig00} 
\end{figure}
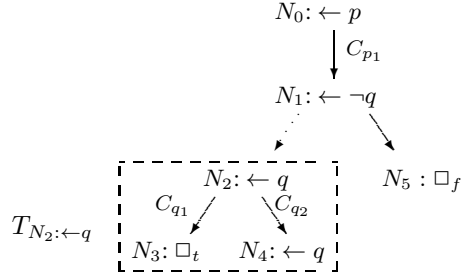 

For simplicity, in the following sections by a derivation or SLDNF-derivation
we refer to a generalized SLDNF-derivation. 
Moreover, for any node $N_i:G_i$ we use 
$L_i^1$ to refer to the selected subgoal in $G_i$.

A derivation step is denoted by $N_i:G_i\Rightarrow_{C,\theta_i} N_{i+1}:G_{i+1}$,
meaning that applying a clause $C$ to $G_i$ produces $N_{i+1}:G_{i+1}$,
where $G_{i+1}$ is the resolvent of $C$ and $G_i$ on $L_i^1$
with the mgu (most general unifier) $\theta_i$.
Here, for a substitution of two variables, $X$ in $L_i^1$ and $Y$ in (the
head of) $C$, we always use $X$ to substitute for $Y$. 
When no confusion would occur, we may omit the mgu $\theta_i$
when writing a derivation step.

\section{A Characterization of Infinite SLDNF-Derivations for Concrete Queries} 
\label{sec2}
In this section, we review the characterization of infinite 
derivations with concrete queries presented in \citeN{shen-tocl}. 

\begin{definition}  
\label{symbol-seq}
Let $T$ be a term or an atom and $S$ be  
a string that consists of all predicate symbols, function  
symbols, constants 
and variables in $T$, which is obtained 
by reading these symbols sequentially from left to  
right. The {\em symbol string} of $T$, denoted 
$S_T$, is the string $S$ with every variable  
replaced by ${\cal X}$. 
\end{definition} 

For instance, let $T_1=a$ and $T_2=f(X,g(X,f(a,Y)))$. 
Then $S_{T_1}=a$ and $S_{T_2}=f{\cal X}g{\cal X}fa{\cal X}$. 
 
\begin{definition}
\label{sub-seq} 
Let $S_{T_1}$ and $S_{T_2}$ be two symbol strings. 
$S_{T_1}$ is a {\em projection} of $S_{T_2}$, denoted  
$S_{T_1}\subseteq_{proj}S_{T_2}$,  
if $S_{T_1}$ is obtained from $S_{T_2}$ by
removing zero or more elements. 
\end{definition} 
 
\begin{definition}  
\label{gvar} 
Let $A_1$ and $A_2$ be two atoms (positive subgoals) with the same predicate symbol. 
$A_1$ is said to {\em loop into} $A_2$, 
denoted $A_1\leadsto_{loop}A_2$, if  
$S_{A_1}\subseteq_{proj}S_{A_2}$. 
Let $N_i:G_i$ and $N_j:G_j$ be two nodes in a  
derivation with $L_i^1\prec_{anc}L_j^1$ and 
$L_i^1 \leadsto_{loop} L_j^1$. 
Then $G_j$ is called a {\em loop goal} of $G_i$. 
\end{definition} 
 
Observe that if $A_1 \leadsto_{loop} A_2$ then $|A_1|\leq |A_2|$, and that 
if $G_3$ is a loop goal of $G_2$ that is a loop goal 
of $G_1$ then $G_3$ is a loop goal of $G_1$. 
Since a logic program has only
a finite number of clauses, an infinite derivation results 
from repeatedly applying the same set of clauses, which leads to either
infinite repetition of selected variant subgoals or  
infinite repetition of selected subgoals with recursive 
increase in term size. By recursive increase of term size 
of a subgoal $A$ from a subgoal $B$ we mean that $A$ is $B$ 
with a few function/constant/variable symbols added  
and possibly with some variables changed to 
different variables. Such crucial dynamic 
characteristics of an infinite derivation  
are captured by loop goals. The following result is proved in \citeN{shen-tocl}.

\begin{theorem} 
\label{th1} 
Let $G_0=\leftarrow Q_0$ be a top goal with $Q_0$ a concrete query.
Any infinite derivation $D$ in $GT_{G_0}$ contains an infinite sequence of
goals $G_0, ..., G_{g_1}, ..., G_{g_2}, ...$  
such that for any $j \geq 1$, $G_{g_{j+1}}$ is a loop goal of $G_{g_j}$. 
\end{theorem}

Put another way, Theorem \ref{th1} states 
that any infinite derivation $D$ in $GT_{G_0}$ is of the form  
\[N_0:G_0\Rightarrow_{C_0} ...\  N_{g_1}:G_{g_1}\Rightarrow_{C_1}...\ 
N_{g_2}:G_{g_2}\Rightarrow_{C_2} ...\  N_{g_3}:G_{g_3}\Rightarrow_{C_3} ...\]  
where for any $j \geq 1$, $G_{g_{j+1}}$ is a loop goal of $G_{g_j}$. 
This provides a necessary and sufficient characterization
of an infinite generalized SLDNF-derivation with a concrete
query. 

\begin{example}
\label{eg0-0}
Consider the following logic program:
\begin{tabbing}
$\qquad$ \= $P_1:$ $\quad$ \= $p(a).$ \`$C_{p_1}$\\
\> \>         $p(f(X))\leftarrow p(X)$. \`$C_{p_2}$
\end{tabbing}
The generalized SLDNF-tree $GT_{\leftarrow p(X)}$ 
for a concrete query $p(X)$ is shown in Figure \ref{fig0-0},
where for simplicity the symbol $\leftarrow$ in each goal is omitted.
Note that $GT_{\leftarrow p(X)}$ has an infinite derivation
\[N_0:p(X)\Rightarrow_{C_{p_2}} N_2:p(X_2)\Rightarrow_{C_{p_2}} N_4:p(X_4)\Rightarrow_{C_{p_2}} ...\]
where for any $j \geq 0$, $G_{2(j+1)}$ is a loop goal of $G_{2j}$.
\end{example}
\begin{figure}[htb]
\begin{center}
\setlength{\unitlength}{3947sp}%
\begingroup\makeatletter\ifx\SetFigFont\undefined%
\gdef\SetFigFont#1#2#3#4#5{%
  \reset@font\fontsize{#1}{#2pt}%
  \fontfamily{#3}\fontseries{#4}\fontshape{#5}%
  \selectfont}%
\fi\endgroup%
\begin{picture}(1200,1446)(2176,-811)
\thinlines
{\color[rgb]{0,0,0}\put(3301,464){\vector( 0,-1){300}}
}%
{\color[rgb]{0,0,0}\put(3301,-61){\vector( 0,-1){300}}
}%
{\color[rgb]{0,0,0}\put(2921,-61){\vector(-2,-1){336}}
}%
{\color[rgb]{0,0,0}\put(2921,464){\vector(-2,-1){336}}
}%
\put(3226,-811){\makebox(0,0)[lb]{\smash{\SetFigFont{12}{14.4}{\rmdefault}{\mddefault}{\updefault}{\color[rgb]{0,0,0}$\vdots$ }%
}}}
\put(2551,-61){\makebox(0,0)[lb]{\smash{\SetFigFont{8}{9.6}{\rmdefault}{\mddefault}{\updefault}{\color[rgb]{0,0,0}$C_{p_1}$}%
}}}
\put(2476,464){\makebox(0,0)[lb]{\smash{\SetFigFont{8}{9.6}{\rmdefault}{\mddefault}{\updefault}{\color[rgb]{0,0,0}$C_{p_1}$}%
}}}
\put(2176,164){\makebox(0,0)[lb]{\smash{\SetFigFont{9}{10.8}{\rmdefault}{\mddefault}{\updefault}{\color[rgb]{0,0,0}$N_1$:  $\Box_t$ }%
}}}
\put(2176,-361){\makebox(0,0)[lb]{\smash{\SetFigFont{9}{10.8}{\rmdefault}{\mddefault}{\updefault}{\color[rgb]{0,0,0}$N_3$:  $\Box_t$ }%
}}}
\put(3376,-211){\makebox(0,0)[lb]{\smash{\SetFigFont{8}{9.6}{\rmdefault}{\mddefault}{\updefault}{\color[rgb]{0,0,0}$\theta_4 = \{X_2/f(X_4)\}$}%
}}}
\put(2851,-511){\makebox(0,0)[lb]{\smash{\SetFigFont{9}{10.8}{\rmdefault}{\mddefault}{\updefault}{\color[rgb]{0,0,0}$N_4$:  $p(X_4)$ }%
}}}
\put(3001,-211){\makebox(0,0)[lb]{\smash{\SetFigFont{8}{9.6}{\rmdefault}{\mddefault}{\updefault}{\color[rgb]{0,0,0}$C_{p_2}$}%
}}}
\put(3376,314){\makebox(0,0)[lb]{\smash{\SetFigFont{8}{9.6}{\rmdefault}{\mddefault}{\updefault}{\color[rgb]{0,0,0}$\theta_2 = \{X/f(X_2)\}$}%
}}}
\put(3001,314){\makebox(0,0)[lb]{\smash{\SetFigFont{8}{9.6}{\rmdefault}{\mddefault}{\updefault}{\color[rgb]{0,0,0}$C_{p_2}$}%
}}}
\put(2851, 14){\makebox(0,0)[lb]{\smash{\SetFigFont{9}{10.8}{\rmdefault}{\mddefault}{\updefault}{\color[rgb]{0,0,0}$N_2$:  $p(X_2)$ }%
}}}
\put(2851,539){\makebox(0,0)[lb]{\smash{\SetFigFont{9}{10.8}{\rmdefault}{\mddefault}{\updefault}{\color[rgb]{0,0,0}$N_0$:  $p(X)$}%
}}}
\end{picture}
\end{center}
\caption{The generalized SLDNF-tree $GT_{\leftarrow p(X)}$ of $P_1$
for a concrete query $p(X)$.}\label{fig0-0} 
\end{figure}
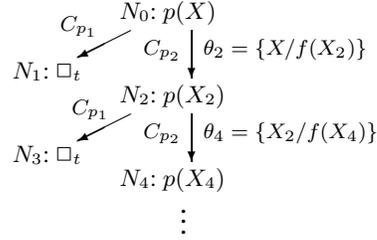 

\section{A Characterization of Infinite SLDNF-Derivations for Moded Queries}
\label{sec3} 

We first define generalized SLDNF-derivations for moded queries
by introducing a framework called moded-query forests. 

\begin{definition} 
\label{mod-tree} 
Let $P$ be a logic program and $Q_0=p({\cal I}_1, ..., {\cal I}_m, T_1, ..., T_n)$
a moded query. The {\em moded-query forest} of $P$ for $Q_0$,
denoted $MF_{Q_0}$, consists of all generalized SLDNF-trees for
$P\cup \{G_0\}$, where $G_0 = \leftarrow p(t_1, ..., t_m, T_1, ..., T_n)$
with each $t_i$ being a ground term from $HU(P)$. 
A ({\em generalized SLDNF-}) {\em derivation for the moded query} $Q_0$
is a derivation in any generalized SLDNF-tree of $MF_{Q_0}$.  
\end{definition}

Therefore, a logic program $P$ terminates for a moded query $Q_0$ 
if and only if there is no infinite derivation for $Q_0$ 
if and only if $MF_{Q_0}$ has no infinite derivation.

\begin{example}
\label{eg0}
Consider the logic program $P_1$ again. 
We have $HU(P_1) = \{a, f(a), f(f(a)), ...\}$. 
Let $p({\cal I})$ be a moded query. 
The moded-query forest $MF_{p({\cal I})}$ consists of generalized SLDNF-trees
$GT_{\leftarrow p(a)}$, $GT_{\leftarrow p(f(a))}$, etc., as shown in Figure \ref{fig0}.
Note that $MF_{p({\cal I})}$ has an infinite number of generalized SLDNF-trees.
However, any individual tree, $GT_{G_0}$ with 
$G_0=\leftarrow p(\underbrace{f(f(...f}_{n \ items}(a)...)))$ ($n\geq 0$), is finite.
$MF_{p({\cal I})}$ contains no infinite derivation,
thus $P_1$ terminates for $p({\cal I})$. 
\end{example}
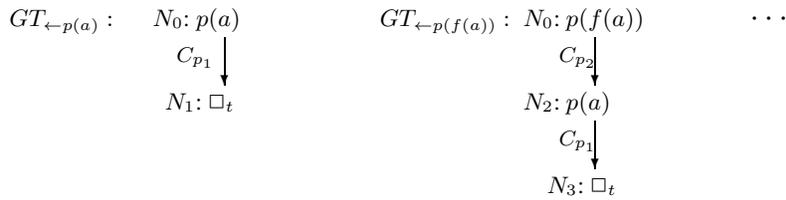
\begin{figure}[htb]
\begin{center}
\setlength{\unitlength}{3947sp}%
\begingroup\makeatletter\ifx\SetFigFont\undefined%
\gdef\SetFigFont#1#2#3#4#5{%
  \reset@font\fontsize{#1}{#2pt}%
  \fontfamily{#3}\fontseries{#4}\fontshape{#5}%
  \selectfont}%
\fi\endgroup%
\begin{picture}(4650,1215)(151,-511)
\thinlines
{\color[rgb]{0,0,0}\put(1501,464){\vector( 0,-1){300}}
}%
{\color[rgb]{0,0,0}\put(3826,464){\vector( 0,-1){300}}
}%
{\color[rgb]{0,0,0}\put(3826,-61){\vector( 0,-1){300}}
}%
\put(3601,314){\makebox(0,0)[lb]{\smash{\SetFigFont{8}{9.6}{\rmdefault}{\mddefault}{\updefault}{\color[rgb]{0,0,0}$C_{p_2}$}%
}}}
\put(3601,-211){\makebox(0,0)[lb]{\smash{\SetFigFont{8}{9.6}{\rmdefault}{\mddefault}{\updefault}{\color[rgb]{0,0,0}$C_{p_1}$}%
}}}
\put(3526,-511){\makebox(0,0)[lb]{\smash{\SetFigFont{9}{10.8}{\rmdefault}{\mddefault}{\updefault}{\color[rgb]{0,0,0}$N_3$:  $\Box_t$ }%
}}}
\put(3376, 14){\makebox(0,0)[lb]{\smash{\SetFigFont{9}{10.8}{\rmdefault}{\mddefault}{\updefault}{\color[rgb]{0,0,0}$N_2$:  $p(a)$ }%
}}}
\put(3376,539){\makebox(0,0)[lb]{\smash{\SetFigFont{9}{10.8}{\rmdefault}{\mddefault}{\updefault}{\color[rgb]{0,0,0}$N_0$:  $p(f(a))$ }%
}}}
\put(1126, 14){\makebox(0,0)[lb]{\smash{\SetFigFont{9}{10.8}{\rmdefault}{\mddefault}{\updefault}{\color[rgb]{0,0,0}$N_1$:  $\Box_t$ }%
}}}
\put(1051,539){\makebox(0,0)[lb]{\smash{\SetFigFont{9}{10.8}{\rmdefault}{\mddefault}{\updefault}{\color[rgb]{0,0,0}$N_0$:  $p(a)$ }%
}}}
\put(1201,314){\makebox(0,0)[lb]{\smash{\SetFigFont{8}{9.6}{\rmdefault}{\mddefault}{\updefault}{\color[rgb]{0,0,0}$C_{p_1}$}%
}}}
\put(4801,539){\makebox(0,0)[lb]{\smash{\SetFigFont{12}{14.4}{\rmdefault}{\mddefault}{\updefault}{\color[rgb]{0,0,0}$\cdots$ }%
}}}
\put(2476,539){\makebox(0,0)[lb]{\smash{\SetFigFont{9}{10.8}{\rmdefault}{\mddefault}{\updefault}{\color[rgb]{0,0,0}$GT_{\leftarrow p(f(a))}:$}%
}}}
\put(151,539){\makebox(0,0)[lb]{\smash{\SetFigFont{9}{10.8}{\rmdefault}{\mddefault}{\updefault}{\color[rgb]{0,0,0}$GT_{\leftarrow p(a)}:$}%
}}}
\end{picture}
\end{center}
\caption{The moded-query forest $MF_{p({\cal I})}$ of $P_1$ for a moded query $p({\cal I})$.}\label{fig0} 
\end{figure} 

In a moded-query forest, all input modes are instantiated
into ground terms in $HU(P)$. When $HU(P)$ is infinite,
the moded-query forest would contain infinitely many
generalized SLDNF-trees. This means that it is infeasible to build 
a moded-query forest to represent the derivations for a moded query.
An alternative yet ideal way is to directly 
apply SLDNF-resolution to evaluate input modes
and build a compact generalized SLDNF-tree for a moded query.
Unfortunately, SLDNF-resolution 
accepts only terms as arguments of a top goal; 
an input mode $\cal I$ is not directly evaluable.

Since an input mode stands for an arbitrary ground term, i.e.
it can be any term from $HU(P)$, during query evaluation 
it can be instantiated to any term except variable
(note that a ground term cannot be substituted by a variable).
This suggests that we may approximate the effect of an input mode ${\cal I}$ 
by treating it as a special (meta-) variable $I$ in the way that in SLDNF-derivations, 
$I$ can be substituted by a constant or function, but cannot be substituted
by an ordinary variable. Therefore, when doing unification of 
a special variable $I$ and a variable $X$, we always substitute $I$ for $X$.  

\begin{definition} 
\label{mod-der} 
Let $P$ be a logic program and $Q_0=p({\cal I}_1, ..., {\cal I}_m, T_1, ..., T_n)$
a moded query. The {\em moded generalized SLDNF-tree} of $P$ for
$Q_0$, denoted $MT_{Q_0}$, is defined 
to be the generalized SLDNF-tree $GT_{G_0}$ for
$P\cup \{G_0\}$, where $G_0 = \leftarrow p(I_1, ..., I_m, T_1, ..., T_n)$
with each $I_i$ being a distinct special variable not occurring in any $T_j$.
The special variables $I_1, ..., I_m$ for the input modes
${\cal I}_1, ..., {\cal I}_m$ are called {\em input mode variables} (or {\em input variables}).
\end{definition}

In a moded generalized SLDNF-tree, an input variable $I$ 
may be substituted by either a constant $t$ or a function $f(\overline{T})$. 
It will not be substituted by any non-input variable.
If $I$ is substituted by $f(\overline{T})$, all variables 
in $\overline{T}$ are also called input variables (thus are treated as special variables). 

In this paper, we do not consider {\em floundering moded} queries; i.e.,
we assume that no negative subgoals containing either ordinary or input variables are
selected at any node of a moded generalized SLDNF-tree.

\begin{definition}
Let $P$ be a logic program,
$Q_0=p({\cal I}_1, ..., {\cal I}_m, T_1, ..., T_n)$ a moded query,
and $G_0 = \leftarrow p(I_1, ..., I_m, T_1, ..., T_n)$.
Let $D$ be a derivation in the moded generalized 
SLDNF-tree $MT_{Q_0}$. A {\em moded instance}
of $D$ is a derivation obtained from $D$ by first 
instantiating all input variables at the root node $N_0:G_0$ 
with an mgu $\theta = \{I_1/t_1, ..., I_m/t_m\}$, where each $t_i\in HU(P)$, 
then passing the instantiation $\theta$ down to the other nodes of $D$.
\end{definition}

\begin{example}
\label{eg1}
Consider the logic program $P_1$ again.
Let $Q_0= p({\cal I})$ be a moded query and $G_0 = \leftarrow p(I)$.
The moded generalized SLDNF-tree $MT_{Q_0}$ is $GT_{G_0}$ as depicted in Figure \ref{fig1},
where all input variables are underlined.
Since $I$ is an input variable, $X_2$ is an input variable, too
(due to the mgu $\theta_2$). For the same reason, 
all $X_{2i}$ are input variables ($i>0$). 

Consider the following infinite derivation $D$ in $MT_{Q_0}$: 
\begin{tabbing}  
$\qquad$ $N_0:p(I)\Rightarrow_{C_{p_2}} N_2:p(X_2) \Rightarrow_{C_{p_2}} N_4:p(X_4) \Rightarrow_{C_{p_2}} \cdots$    
\end{tabbing}
By instantiating the input variable $I$ at $N_0$
with different ground terms from $HU(P_1)$ and passing 
the instantiation $\theta$ down to the other nodes of $D$,
we can obtain different moded instances from $D$.
For example, instantiating $I$ to $a$ (i.e. $\theta = \{I/a\}$) yields the moded instance
\begin{tabbing} 
$\qquad$ $N_0:p(a)$
\end{tabbing}
Instantiating $I$ to $f(a)$ (i.e. $\theta = \{I/f(a)\}$) yields the moded instance
\begin{tabbing} 
$\qquad$ $N_0:p(f(a))\Rightarrow_{C_{p_2}} N_2:p(a)$
\end{tabbing}
And, instantiating $I$ to $f(f(a))$ (i.e. $\theta = \{I/f(f(a))\}$) yields the moded instance
\begin{tabbing} 
$\qquad$ $N_0:p(f(f(a)))\Rightarrow_{C_{p_2}} N_2:p(f(a)) \Rightarrow_{C_{p_2}} N_4:p(a)$
\end{tabbing}
\end{example}
\begin{figure}[htb]
\begin{center}
\setlength{\unitlength}{3947sp}%
\begingroup\makeatletter\ifx\SetFigFont\undefined%
\gdef\SetFigFont#1#2#3#4#5{%
  \reset@font\fontsize{#1}{#2pt}%
  \fontfamily{#3}\fontseries{#4}\fontshape{#5}%
  \selectfont}%
\fi\endgroup%
\begin{picture}(1200,1446)(2176,-811)
\thinlines
{\color[rgb]{0,0,0}\put(3301,464){\vector( 0,-1){300}}
}%
{\color[rgb]{0,0,0}\put(3301,-61){\vector( 0,-1){300}}
}%
{\color[rgb]{0,0,0}\put(2921,-61){\vector(-2,-1){336}}
}%
{\color[rgb]{0,0,0}\put(2921,464){\vector(-2,-1){336}}
}%
\put(3226,-811){\makebox(0,0)[lb]{\smash{\SetFigFont{12}{14.4}{\rmdefault}{\mddefault}{\updefault}{\color[rgb]{0,0,0}$\vdots$ }%
}}}
\put(2551,-61){\makebox(0,0)[lb]{\smash{\SetFigFont{8}{9.6}{\rmdefault}{\mddefault}{\updefault}{\color[rgb]{0,0,0}$C_{p_1}$}%
}}}
\put(2476,464){\makebox(0,0)[lb]{\smash{\SetFigFont{8}{9.6}{\rmdefault}{\mddefault}{\updefault}{\color[rgb]{0,0,0}$C_{p_1}$}%
}}}
\put(2176,164){\makebox(0,0)[lb]{\smash{\SetFigFont{9}{10.8}{\rmdefault}{\mddefault}{\updefault}{\color[rgb]{0,0,0}$N_1$:  $\Box_t$ }%
}}}
\put(2176,-361){\makebox(0,0)[lb]{\smash{\SetFigFont{9}{10.8}{\rmdefault}{\mddefault}{\updefault}{\color[rgb]{0,0,0}$N_3$:  $\Box_t$ }%
}}}
\put(3376,-211){\makebox(0,0)[lb]{\smash{\SetFigFont{8}{9.6}{\rmdefault}{\mddefault}{\updefault}{\color[rgb]{0,0,0}$\theta_4 = \{\underline{X_2}/f(X_4)\}$}%
}}}
\put(2851,-511){\makebox(0,0)[lb]{\smash{\SetFigFont{9}{10.8}{\rmdefault}{\mddefault}{\updefault}{\color[rgb]{0,0,0}$N_4$:  $p(\underline{X_4})$ }%
}}}
\put(3001,-211){\makebox(0,0)[lb]{\smash{\SetFigFont{8}{9.6}{\rmdefault}{\mddefault}{\updefault}{\color[rgb]{0,0,0}$C_{p_2}$}%
}}}
\put(3376,314){\makebox(0,0)[lb]{\smash{\SetFigFont{8}{9.6}{\rmdefault}{\mddefault}{\updefault}{\color[rgb]{0,0,0}$\theta_2 = \{\underline{I}/f(X_2)\}$}%
}}}
\put(3001,314){\makebox(0,0)[lb]{\smash{\SetFigFont{8}{9.6}{\rmdefault}{\mddefault}{\updefault}{\color[rgb]{0,0,0}$C_{p_2}$}%
}}}
\put(2851, 14){\makebox(0,0)[lb]{\smash{\SetFigFont{9}{10.8}{\rmdefault}{\mddefault}{\updefault}{\color[rgb]{0,0,0}$N_2$:  $p(\underline{X_2})$ }%
}}}
\put(2851,539){\makebox(0,0)[lb]{\smash{\SetFigFont{9}{10.8}{\rmdefault}{\mddefault}{\updefault}{\color[rgb]{0,0,0}$N_0$:  $p(\underline{I})$}%
}}}
\end{picture}
\end{center}
\caption{The moded generalized SLDNF-tree $MT_{p({\cal I})}$ 
of $P_1$ for a moded query $p({\cal I})$.}\label{fig1} 
\end{figure}
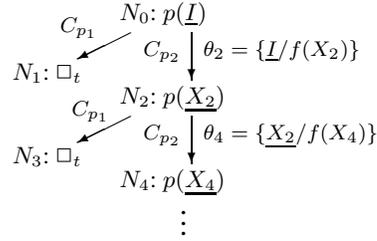

Observe that a moded instance of a derivation $D$ in $MT_{Q_0}$ is a derivation 
in $GT_{G_0\theta}$, where $G_0\theta = \leftarrow p(t_1, ..., t_m, T_1, ..., T_n)$ with
each $t_i$ being a ground term from $HU(P)$. 
By Definition \ref{mod-tree}, $GT_{G_0\theta}$ is in the moded-query forest $MF_{Q_0}$.
This means that any moded instance of a derivation in $MT_{Q_0}$ is a derivation for $Q_0$
in $MF_{Q_0}$. For instance, all moded instances illustrated in Example \ref{eg1} are derivations 
in the moded-query forest $MF_{Q_0}$ of Figure~\ref{fig0}.  

\begin{theorem}
\label{th2}
Let $MF_{Q_0}$ and $MT_{Q_0}$ be the moded-query forest and the
moded generalized SLDNF-tree of $P$ for $Q_0$, respectively.
If $MF_{Q_0}$ has an infinite derivation $D'$,
$MT_{Q_0}$ has an infinite derivation $D$ with 
$D'$ as a moded instance. 
\end{theorem}

\begin{proof}
Let $Q_0=p({\cal I}_1, ..., {\cal I}_m, T_1, ..., T_n)$. 
Then, the root node of $D'$ is $N_0:\leftarrow p(t_1, ..., t_m,$ $T_1, ..., T_n)$
with each $t_i\in HU(P)$, and the root node of $MT_{Q_0}$ is
$N_0 : \leftarrow p(I_1, ..., I_m,$ $T_1, ..., T_n)$
with each $I_i$ being an input variable not occurring in any $T_j$.
Note that the former is an instance of the latter with the mgu 
$\theta = \{I_1/t_1, ..., I_m/t_m\}$. 
Let $D'$ be of the form
\begin{tabbing}  
$\qquad$ $N_0:\leftarrow p(t_1, ..., t_m, T_1, ..., T_n)
\Rightarrow_{C_0} N_1:G_1' \cdots \Rightarrow_{C_i} N_{i+1}:G_{i+1}' \cdots$    
\end{tabbing}
$MT_{Q_0}$ must have a derivation $D$ of the form
\begin{tabbing}  
$\qquad$ $N_0:\leftarrow p(I_1, ..., I_m, T_1, ..., T_n)
\Rightarrow_{C_0} N_1:G_1 \cdots \Rightarrow_{C_i} N_{i+1}:G_{i+1} \cdots$    
\end{tabbing}
such that each $G_i' = G_i\theta$, since for any $i\geq 0$
and any clause $C_i$ in $P$, if $G_i'$ can unify with $C_i$,
so can $G_i$ with $C_i$. Note that when the selected subgoal at some $G_i'$
is a negative ground literal, by the assumption that $Q_0$ is non-floundering,
we have the same selected literal at $G_i$.
We then have the proof.  
\end{proof}

Our goal is to establish a characterization of 
infinite derivations for a moded query such that 
the converse of Theorem \ref{th2} is true under some conditions.

Consider the infinite derivation in Figure \ref{fig1} again. The input variable $I$
is substituted by $f(X_2)$; $X_2$ is then substituted by $f(X_4)$, \ldots 
This produces an infinite chain of substitutions for $I$ 
of the form $I/f(X_2), X_2/f(X_4),$ \ldots The following lemma shows
that infinite derivations containing such an infinite chain of substitutions 
have no infinite moded instances.
 
\begin{lemma} 
\label{lem1}
If a derivation $D$ in a moded generalized SLDNF-tree 
$MT_{Q_0}$ is infinite but none of its moded instances
is infinite, then there is an input variable $I$ such that $D$ contains
an infinite chain of substitutions for $I$ of the form
\begin{eqnarray}
\label{ins-chain}
I/f_1(...,Y_1,...), ..., Y_1/f_2(...,Y_2,...), ...,Y_{i-1}/f_i(...,Y_i,...), ...
\end{eqnarray}
(some $f_i$s would be the same). 
\end{lemma}

\begin{proof}
We distinguish four types of substitution chains for an input variable $I$ in $D$:
\begin{enumerate}
\item
\label{sub1}
$X_1/I, ..., X_m/I$ or $X_1/I, ..., X_i/I,$ \ldots That is, $I$ is
never substituted by any terms.

\item
\label{sub2}
$X_1/I, ..., X_m/I, I/t$ where $t$ is a ground term. That is, $I$ is substituted by
a ground term.

\item
\label{sub3}
$X_1/I, ...,X_m/I, I/f_1(...,Y_1,...), ..., Y_1/f_2(...,Y_2,...), ..., Y_{n-1}/f_n(...,Y_n,...), ... \ $ 
where
$f_n($ $...,Y_n,...)$ is the last non-ground function in the substitution chain for $I$ in $D$.
In this case, $I$ is recursively substituted by a finite number of functions.

\item
\label{sub4}
$X_1/I, ..., X_m/I, I/f_1(...,Y_1,...), ..., Y_1/f_2(...,Y_2,...), ..., Y_{i-1}/f_i(...,Y_i,...),$ \ldots
In this case, $I$ is recursively substituted by an infinite number of functions. 
\end{enumerate}
For type \ref{sub1}, $D$ retains its infinite extension for whatever ground term we replace $I$ with.
For type \ref{sub2}, $D$ retains its infinite extension when we use $t$ to replace $I$.
To sum up, for any input variable $I$ whose substitution chain is of 
type \ref{sub1} or of type \ref{sub2}, there is a ground term $t$ such that 
replacing $I$ with $t$ does not affect the infinite extension of $D$.
In this case, replacing $I$ in $D$ with $t$ leads to an infinite derivation less general than $D$. 
 
For type \ref{sub3}, note that all variables appearing in the $f_i(.)$s are
input variables. Since $f_n(...,Y_n,...)$ is the last non-ground 
function in the substitution chain for $I$ in $D$,
the substitution chain for every variable $Y_n$ in $f_n(...,Y_n,...)$ 
is either of type \ref{sub1} or of type \ref{sub2}.
Therefore, we can replace each $Y_n$
with an appropriate ground term $t_n$ without affecting the infinite extension of $D$.
After this replacement, $D$ becomes $D_n$ and $f_n(...,Y_n,...)$ 
becomes a ground term $f_n(...,t_n,...)$.
Now $f_{n-1}(...,Y_{n-1},...)$ is the last non-ground
function in the substitution chain for $I$ in $D_n$.
Repeating the above replacement recursively, we will obtain an infinite derivation
$D_1$, which is $D$ with all variables in the $f_i(.)$s replaced with a ground term.
Assume $f_1(...,Y_1,...)$ becomes a ground term $t$ in $D_1$.
Then the substitution chain for $I$ in $D_1$ is of type \ref{sub2}. 
So replacing $I$ with $t$ in $D_1$ leads to an infinite derivation $D_0$.

The above constructive proof shows that if the substitution chains for all input variables in $D$
are of type \ref{sub1}, \ref{sub2} or \ref{sub3}, then $D$ must have an infinite moded instance.
Since $D$ has no infinite moded instance, there must exist an input variable $I$ 
whose substitution chain in $D$ is of type \ref{sub4}. That is,
$I$ is recursively substituted by an infinite number of functions.
Note that some $f_i$s would be the same because a logic program has only a finite number
of function symbols. This concludes the proof.
\end{proof}

We are ready to introduce the following principal result. 

\begin{theorem} 
\label{th-main} 
Let $MF_{Q_0}$ and $MT_{Q_0}$ be the moded-query forest and the
moded generalized SLDNF-tree of $P$ for $Q_0$, respectively.
$MF_{Q_0}$ has an infinite derivation 
if and only if $MT_{Q_0}$ has an infinite derivation $D$ of the form 
\begin{eqnarray} 
\label{eq2}
N_0:G_0\Rightarrow_{C_0} ...\  N_{g_1}:G_{g_1}\Rightarrow_{C_1}...\ 
N_{g_2}:G_{g_2}\Rightarrow_{C_2} ...\ N_{g_3}:G_{g_3}\Rightarrow_{C_3} ... 
\end{eqnarray} 
where (i) for any $j \geq 1$, $G_{g_{j+1}}$ is a loop goal of $G_{g_j}$,
and (ii) for no input variable $I$, $D$ contains 
an infinite chain of substitutions for $I$ of form (\ref{ins-chain}).
\end{theorem}

\begin{proof}
($\Longrightarrow$) Assume $MF_{Q_0}$ has an infinite derivation $D'$.
By Theorem \ref{th2}, $GT_{G_0}$ has an infinite derivation $D$
with $D'$ as a moded instance. By Theorem \ref{th1}, $D$ is of form (\ref{eq2})
and satisfies condition (i). 
 
Assume, on the contrary, that $D$ does not satisfy condition (ii).
That is, for some input variable $I$, $D$ contains 
an infinite chain of substitutions for $I$ of the form
\[I/f_1(...,Y_1,...), ..., Y_1/f_2(...,Y_2,...), ..., Y_{i-1}/f_i(...,Y_i,...), ...\] 
Note that for whatever ground term $t$ we assign to $I$, this chain can be instantiated
at most as long in length as the following one: 
\[t/f_1(...,t_1,...), ..., t_1/f_2(...,t_2,...), ..., t_k/f_{k+1}(...,Y_{k+1},...)\] 
where $k = |t|$, $t_i$s are ground terms and $|t_k| = 1$. This means that replacing $I$ with any ground term $t$
leads to a finite moded instance of $D$. Therefore, $D$ has no infinite moded instance in $MF_{Q_0}$,
a contradiction.

($\Longleftarrow$) Assume, on the contrary, that $MF_{Q_0}$ has no infinite 
derivation. By Lemma \ref{lem1}, we reach a contradiction to condition (ii).  
\end{proof}

Theorem \ref{th-main} provides a necessary and sufficient characterization
of an infinite generalized SLDNF-derivation for a moded query.
Note that it coincides with Theorem \ref{th1} when $Q_0$ is a concrete
query, where $MF_{Q_0} = MT_{Q_0}$ and condition (ii) is always true.
 
The following corollary is immediate to this theorem.

\begin{corollary}
\label{cor-main}
A logic program $P$ terminates for a moded query $Q_0$ 
if and only if the moded generalized SLDNF-tree 
$MT_{Q_0}$ has no infinite derivation of form (\ref{eq2}) 
satisfying conditions (i) and (ii) of Theorem \ref{th-main}.
\end{corollary}

We use simple yet typical examples to illustrate 
the proposed characterization of infinite
SLDNF-derivations with moded queries.

\begin{example}
\label{eg0-1}
Consider the moded generalized SLDNF-tree $MT_{Q_0}$ in Figure \ref{fig1}.
It has only one infinite derivation, which satisfies condition (i)
of Theorem \ref{th-main} where for each $j\geq 0$, 
$N_{g_j}$ in Theorem \ref{th-main} corresponds to $N_{2j}$ in Figure \ref{fig1}.
However, the chain of substitutions for $I$ in this derivation violates condition (ii).
This means that $MF_{Q_0}$ contains no infinite derivations; 
therefore, there is no infinite derivation for 
the moded query $p({\cal I})$. As a result, $P_1$ terminates for $p({\cal I})$. 
\end{example}

\begin{example} 
\label{eg2}  
Consider the {\em append} program:
\begin{tabbing} 
$\qquad$ \= $P_2:$ $\quad$ \= $append([],X,X)$. \`$C_{a_1}$\\ 
\> \> $append([X|Y],U,[X|Z])\leftarrow append(Y,U,Z)$.       \`$C_{a_2}$  
\end{tabbing} 
Let us choose the three simplest moded queries:  
\begin{tabbing} 
$\qquad\qquad$ \= $Q_0^1= append({\cal I}, V_2, V_3)$,\\ 
\>                $Q_0^2= append(V_1, {\cal I}, V_3)$,\\ 
\>                $Q_0^3= append(V_1, V_2, {\cal I})$.
\end{tabbing} 
Since applying clause $C_{a_1}$ produces only leaf nodes,
for simplicity we ignore it when depicting moded generalized
SLDNF-trees. The three moded generalized
SLDNF-trees $MT_{Q_0^1}$, $MT_{Q_0^2}$ and $MT_{Q_0^3}$
are shown in Figures \ref{fig2} (a), (b) and (c), respectively. Note that all the derivations
are infinite and satisfy condition (i) of Theorem \ref{th-main}, where for each $j\geq 0$, 
$N_{g_j}$ in Theorem \ref{th-main} corresponds to $N_j$ in Figure \ref{fig2}. 
Apparently, the chains of substitutions for $I$ in 
the derivations of $MT_{Q_0^1}$ and $MT_{Q_0^3}$ violate condition (ii) of Theorem \ref{th-main}.
$MF_{Q_0^1}$ and $MF_{Q_0^3}$ contain no infinite derivation and thus 
there exists no infinite derivation for the moded queries $Q_0^1$ and $Q_0^3$.
Therefore, $P_2$ terminates for $Q_0^1$ and $Q_0^3$.
However, the derivation in $MT_{Q_0^2}$ satisfies condition (ii), thus 
there exist infinite derivations for the moded query $Q_0^2$.  
$P_2$ does not terminate for $Q_0^2$.
\end{example} 
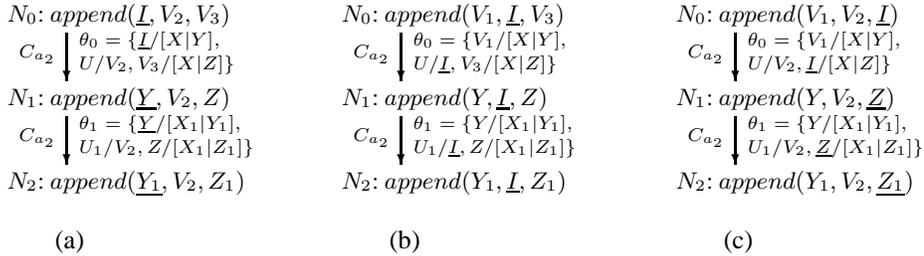
\begin{figure}[htb]
\begin{tabbing}
\setlength{\unitlength}{3947sp}%
\begingroup\makeatletter\ifx\SetFigFont\undefined%
\gdef\SetFigFont#1#2#3#4#5{%
  \reset@font\fontsize{#1}{#2pt}%
  \fontfamily{#3}\fontseries{#4}\fontshape{#5}%
  \selectfont}%
\fi\endgroup%
\begin{picture}(4721,1561)(826,-851)
\put(1126,-811){\makebox(0,0)[lb]{\smash{{\SetFigFont{10}{12.0}{\rmdefault}{\mddefault}{\updefault}{\color[rgb]{0,0,0}(a)}%
}}}}
\put(3226,-811){\makebox(0,0)[lb]{\smash{{\SetFigFont{10}{12.0}{\rmdefault}{\mddefault}{\updefault}{\color[rgb]{0,0,0}(b)}%
}}}}
\put(5326,-811){\makebox(0,0)[lb]{\smash{{\SetFigFont{10}{12.0}{\rmdefault}{\mddefault}{\updefault}{\color[rgb]{0,0,0}(c)}%
}}}}
\thinlines
{\color[rgb]{0,0,0}\put(1201,539){\vector( 0,-1){300}}
}%
{\color[rgb]{0,0,0}\put(1201, 14){\vector( 0,-1){300}}
}%
\put(826, 89){\makebox(0,0)[lb]{\smash{{\SetFigFont{9}{10.8}{\rmdefault}{\mddefault}{\updefault}{\color[rgb]{0,0,0}$N_1$:  $append(\underline{Y},V_2,Z)$ }%
}}}}
\put(826,-436){\makebox(0,0)[lb]{\smash{{\SetFigFont{9}{10.8}{\rmdefault}{\mddefault}{\updefault}{\color[rgb]{0,0,0}$N_2$:  $append(\underline{Y_1},V_2,Z_1)$ }%
}}}}
\put(826,614){\makebox(0,0)[lb]{\smash{{\SetFigFont{9}{10.8}{\rmdefault}{\mddefault}{\updefault}{\color[rgb]{0,0,0}$N_0$:  $append(\underline{I}, V_2, V_3)$}%
}}}}
\put(901,389){\makebox(0,0)[lb]{\smash{{\SetFigFont{7}{8.4}{\rmdefault}{\mddefault}{\updefault}{\color[rgb]{0,0,0}$C_{a_2}$}%
}}}}
\put(901,-136){\makebox(0,0)[lb]{\smash{{\SetFigFont{7}{8.4}{\rmdefault}{\mddefault}{\updefault}{\color[rgb]{0,0,0}$C_{a_2}$}%
}}}}
\put(1276,-61){\makebox(0,0)[lb]{\smash{{\SetFigFont{7}{8.4}{\rmdefault}{\mddefault}{\updefault}{\color[rgb]{0,0,0}$\theta_1 = \{\underline{Y}/[X_1|Y_1],$}%
}}}}
\put(1276,464){\makebox(0,0)[lb]{\smash{{\SetFigFont{7}{8.4}{\rmdefault}{\mddefault}{\updefault}{\color[rgb]{0,0,0}$\theta_0 = \{\underline{I}/[X|Y],$}%
}}}}
\put(1276,-211){\makebox(0,0)[lb]{\smash{{\SetFigFont{7}{8.4}{\rmdefault}{\mddefault}{\updefault}{\color[rgb]{0,0,0}$U_1/V_2,Z/[X_1|Z_1]\}$}%
}}}}
\put(1276,314){\makebox(0,0)[lb]{\smash{{\SetFigFont{7}{8.4}{\rmdefault}{\mddefault}{\updefault}{\color[rgb]{0,0,0}$U/V_2,V_3/[X|Z]\}$}%
}}}}
{\color[rgb]{0,0,0}\put(3301,539){\vector( 0,-1){300}}
}%
{\color[rgb]{0,0,0}\put(3301, 14){\vector( 0,-1){300}}
}%
\put(2926, 89){\makebox(0,0)[lb]{\smash{{\SetFigFont{9}{10.8}{\rmdefault}{\mddefault}{\updefault}{\color[rgb]{0,0,0}$N_1$:  $append(Y,\underline{I},Z)$ }%
}}}}
\put(2926,-436){\makebox(0,0)[lb]{\smash{{\SetFigFont{9}{10.8}{\rmdefault}{\mddefault}{\updefault}{\color[rgb]{0,0,0}$N_2$:  $append(Y_1,\underline{I},Z_1)$ }%
}}}}
\put(2926,614){\makebox(0,0)[lb]{\smash{{\SetFigFont{9}{10.8}{\rmdefault}{\mddefault}{\updefault}{\color[rgb]{0,0,0}$N_0$:  $append(V_1, \underline{I}, V_3)$}%
}}}}
\put(3001,389){\makebox(0,0)[lb]{\smash{{\SetFigFont{7}{8.4}{\rmdefault}{\mddefault}{\updefault}{\color[rgb]{0,0,0}$C_{a_2}$}%
}}}}
\put(3001,-136){\makebox(0,0)[lb]{\smash{{\SetFigFont{7}{8.4}{\rmdefault}{\mddefault}{\updefault}{\color[rgb]{0,0,0}$C_{a_2}$}%
}}}}
\put(3376,-61){\makebox(0,0)[lb]{\smash{{\SetFigFont{7}{8.4}{\rmdefault}{\mddefault}{\updefault}{\color[rgb]{0,0,0}$\theta_1 = \{Y/[X_1|Y_1],$}%
}}}}
\put(3376,464){\makebox(0,0)[lb]{\smash{{\SetFigFont{7}{8.4}{\rmdefault}{\mddefault}{\updefault}{\color[rgb]{0,0,0}$\theta_0 = \{V_1/[X|Y],$}%
}}}}
\put(3376,-211){\makebox(0,0)[lb]{\smash{{\SetFigFont{7}{8.4}{\rmdefault}{\mddefault}{\updefault}{\color[rgb]{0,0,0}$U_1/\underline{I},Z/[X_1|Z_1]\}$}%
}}}}
\put(3376,314){\makebox(0,0)[lb]{\smash{{\SetFigFont{7}{8.4}{\rmdefault}{\mddefault}{\updefault}{\color[rgb]{0,0,0}$U/\underline{I},V_3/[X|Z]\}$}%
}}}}
{\color[rgb]{0,0,0}\put(5401,539){\vector( 0,-1){300}}
}%
{\color[rgb]{0,0,0}\put(5401, 14){\vector( 0,-1){300}}
}%
\put(5026, 89){\makebox(0,0)[lb]{\smash{{\SetFigFont{9}{10.8}{\rmdefault}{\mddefault}{\updefault}{\color[rgb]{0,0,0}$N_1$:  $append(Y,V_2,\underline{Z})$ }%
}}}}
\put(5026,-436){\makebox(0,0)[lb]{\smash{{\SetFigFont{9}{10.8}{\rmdefault}{\mddefault}{\updefault}{\color[rgb]{0,0,0}$N_2$:  $append(Y_1,V_2,\underline{Z_1})$ }%
}}}}
\put(5026,614){\makebox(0,0)[lb]{\smash{{\SetFigFont{9}{10.8}{\rmdefault}{\mddefault}{\updefault}{\color[rgb]{0,0,0}$N_0$:  $append(V_1, V_2, \underline{I})$}%
}}}}
\put(5101,389){\makebox(0,0)[lb]{\smash{{\SetFigFont{7}{8.4}{\rmdefault}{\mddefault}{\updefault}{\color[rgb]{0,0,0}$C_{a_2}$}%
}}}}
\put(5101,-136){\makebox(0,0)[lb]{\smash{{\SetFigFont{7}{8.4}{\rmdefault}{\mddefault}{\updefault}{\color[rgb]{0,0,0}$C_{a_2}$}%
}}}}
\put(5476,-61){\makebox(0,0)[lb]{\smash{{\SetFigFont{7}{8.4}{\rmdefault}{\mddefault}{\updefault}{\color[rgb]{0,0,0}$\theta_1 = \{Y/[X_1|Y_1],$}%
}}}}
\put(5476,464){\makebox(0,0)[lb]{\smash{{\SetFigFont{7}{8.4}{\rmdefault}{\mddefault}{\updefault}{\color[rgb]{0,0,0}$\theta_0 = \{V_1/[X|Y],$}%
}}}}
\put(5476,-211){\makebox(0,0)[lb]{\smash{{\SetFigFont{7}{8.4}{\rmdefault}{\mddefault}{\updefault}{\color[rgb]{0,0,0}$U_1/V_2,\underline{Z}/[X_1|Z_1]\}$}%
}}}}
\put(5476,314){\makebox(0,0)[lb]{\smash{{\SetFigFont{7}{8.4}{\rmdefault}{\mddefault}{\updefault}{\color[rgb]{0,0,0}$U/V_2,\underline{I}/[X|Z]\}$}%
}}}}
\end{picture}%
\end{tabbing} 
\caption{(a) $MT_{Q_0^1}$, (b) $MT_{Q_0^2}$, and (c) $MT_{Q_0^3}$.}\label{fig2} 
\end{figure} 

Let $pred(P)$ be the set of predicate symbols in $P$.
Define
\begin{tabbing}
$\qquad\qquad MQ(P)$ \= $ = \{p(\overline{T}) \mid p$ is an
             $n$-ary predicate symbol in $pred(P)$,\\
\>           and $\overline{T}$ consists of $m>0$ input modes and $n-m$ variables$\}$.
\end{tabbing}
Note that $MQ(P)$ contains all most general moded queries
of $P$ in the sense that any moded query of $P$ is an instance of
some query in $MQ(P)$. 
Since $pred(P)$ is finite, $MQ(P)$ is finite.
Therefore, it is immediate that $P$ terminates for all moded queries 
if and only if it terminates for each moded query in $MQ(P)$.

\begin{theorem} 
\label{th-inc}
Let $Q_1 = p(\overline{T_1})$ and $Q_2 = p(\overline{T_2})$ be two moded queries in $MQ(P)$,
where all input modes of $Q_1$ occur in $Q_2$.  
If there is no infinite derivation for $Q_1$, 
there is no infinite derivation for $Q_2$. 
\end{theorem} 

\begin{proof}
Note that we consider only non-floundering queries
by assuming that no negative subgoals containing either ordinary or input variables are
selected at any node of a moded generalized SLDNF-tree.
Then, for any concrete query $Q$,
that there is no infinite derivation for $Q$ 
implies there is no infinite derivation for any instance of $Q$.

For ease of presentation, let
$Q_1 = p({\cal I}_1, ..., {\cal I}_l, X_{l+1},...,X_n)$ 
and $Q_2 = p({\cal I}_1, ..., {\cal I}_m,$ $X_{m+1},...,X_n)$ with $l<m$. 

Assume that there is no infinite derivation for $Q_1$.
Then, there is no infinite derivation for any query 
$Q = p(t_1, ..., t_l, X_{l+1},...,X_n)$, where each $t_i$
is a ground term from $HU(P)$. Then, there is no infinite derivation for any query 
$Q' = p(t_1, ..., t_l, s_{l+1}, ..., s_m, X_{m+1},...,$ $X_n)$, where each $t_i$
is a ground term from $HU(P)$ and each $s_i$ an instance of $X_i$. 
Since all $X_i$s are variables, there is no infinite derivation for any query 
$Q'' = p(t_1, ..., t_l, t_{l+1}, ..., t_m,$ $X_{m+1},...,X_n)$, where each $t_i$
is a ground term from $HU(P)$. That is,
there is no infinite derivation for $Q_2$. 
\end{proof}

Applying this theorem, we can conclude that $P_2$ in Example \ref{eg2}
terminates for all moded queries in $MQ(P_2)$ except $Q_0^2$.

\section{An Algorithm for Predicting Termination of Logic Programs}
\label{sec4}

We develop an algorithm for predicting termination of logic programs
based on the necessary and sufficient characterization
of an infinite generalized SLDNF-derivation (Theorem \ref{th-main}
and Corollary \ref{cor-main}). 
We begin by introducing a loop checking mechanism. 

A loop checking mechanism, or more formally a {\em loop check} \cite{BAK91},
defines conditions for us to cut a possibly infinite derivation at some node.
By cutting a derivation at a node $N$ we mean removing all descendants of $N$. 
Informally, a loop check is said to be {\em weakly sound} if for any
generalized SLDNF-tree $GT_{G_0}$, $GT_{G_0}$ having a success derivation before cut
implies it has a success derivation after cut; it is said to be {\em complete}
if it cuts all infinite derivations in $GT_{G_0}$. 
An ideal loop check cuts all infinite derivations while retaining success
derivations. Unfortunately, as shown by \citeN{BAK91},
there exists no loop check that is both weakly sound and complete. In this paper,
we focus on complete loop checks, because we want to apply them to
predict termination of logic programs.

\begin{definition}
\label{mq-check} 
Given a repetition number $r\geq 2$, {\em LP-check} is defined as follows:
Any derivation $D$ in a generalized SLDNF-tree is cut at 
a node $N_{g_r}$ if $D$ has a prefix of the form
\begin{eqnarray} 
\label{eq3}
N_0:G_0\Rightarrow_{C_0} ...\ N_{g_1}:G_{g_1}\Rightarrow_{C_k} ...\   
N_{g_2}:G_{g_2}\Rightarrow_{C_k} ...\ N_{g_r}:G_{g_r}\Rightarrow_{C_k} 
\end{eqnarray} 
such that 
(a) for any $j < r$, $G_{g_{j+1}}$ is a loop goal of $G_{g_j}$, 
and (b) for all $j \leq r$, the clause $C_k$ applied to $G_{g_j}$ is the same.
$C_k$ is then called a {\em looping clause}.  
\end{definition} 

LP-check predicts infinite derivations from 
prefixes of derivations based on the characterization of 
Theorem \ref{th1} (or condition (i) of Theorem \ref{th-main}). 
The repetition number $r$ specifies the minimum
number of loop goals appearing in the prefixes. 
It appears not appropriate to choose $r<2$, as that may lead to
many finite derivations being wrongly cut.
Although there is no mathematical mechanism available 
for choosing this repetition number (since the termination
problem is undecidable), 
in many situations it suffices 
to choose $r=3$ for a correct prediction of infinite derivations.
For instance, choosing $r=3$ we are able to obtain correct predictions
for all benchmark programs of the Termination Competition 2007
(see Section \ref{implementation}). 

LP-check applies to any generalized SLDNF-trees including
moded generalized SLDNF-trees.

\begin{theorem} 
\label{check-comp} 
LP-check is a complete loop check.
\end{theorem} 

\begin{proof} 
Let $D$ be an infinite derivation in $GT_{G_0}$.
By Theorem \ref{th1}, $D$ is of the form
\[N_0:G_0\Rightarrow_{C_0} ...\ N_{f_1}:G_{f_1}\Rightarrow_{C_1}   
...\ N_{f_2}:G_{f_2}\Rightarrow_{C_2} ...\] 
such that for any $i \geq 1$, $G_{f_{i+1}}$ is a loop goal of $G_{f_i}$.
Since a logic program has only a finite number of clauses,
there must be a (looping) clause $C_k$ being repeatedly applied at infinitely many nodes 
$N_{g_1}:G_{g_1}, N_{g_2}:G_{g_2}, \cdots$ 
where for each $j \geq 1$, $g_j\in \{f_1, f_2, ...\}$.
Then for any $r>0$, $D$ has a partial derivation of form (\ref{eq3}). So $D$ will be cut
at node $N_{g_r}:G_{g_r}$. This shows
that any infinite derivation can be cut by LP-check. 
That is, LP-check is a complete loop check. 
\end{proof}

\begin{example}
\label{eg1-1}
Let us choose $r=3$ and consider the infinite derivation $D$ in 
Figure \ref{fig1}. $p(X_4)$ at $N_4$ is a loop goal of $p(X_2)$ at $N_2$ 
that is a loop goal of $p(I)$ at $N_0$. Moreover, the same clause
$C_{p_2}$ is applied at the three nodes. $D$  
satisfies the conditions of LP-check and is cut at node $N_4$.
\end{example}

Recall that to prove that a logic program $P$ terminates
for a moded query $Q_0 = p({\cal I}_1, ...,$ ${\cal I}_m,T_1, ..., T_n)$ 
is to prove that 
$P$ terminates for any query $p(t_1, ..., t_m,$ $T_2, ..., T_n)$,
where each $t_i$ is a ground term. This can be reformulated in
terms of a moded-query forest; that is, $P$ terminates for $Q_0$ if and only if
$MF_{Q_0}$ has no infinite derivation. Then, Corollary \ref{cor-main}
shows that $P$ terminates for $Q_0$ if and only if
the moded generalized SLDNF-tree $MT_{Q_0}$ 
has no infinite derivation $D$ of form (\ref{eq2}) satisfying
the two conditions (i) and (ii) of Theorem~\ref{th-main}. 
Although this characterization cannot be 
directly used for automated termination test because it requires generating
infinite derivations in $MT_{Q_0}$, it can be used
along with LP-check to predict termination, as LP-check is able to guess
if a partial derivation would extend to an infinite one.
Before describing our prediction algorithm 
with this idea, we introduce one more condition 
following Definition \ref{mq-check}.

\begin{definition}
\label{term-dec} 
Let $D$ be a derivation with a prefix of form (\ref{eq3}).
The prefix of $D$ is said to have the {\em term-size decrease} property
if for any $i$ with $0<i<r$, there is a substitution $X/f(...Y...)$ 
between $N_{g_i}$ and $N_{g_{i+1}}$, where $X$ is an input 
variable and $Y$ (an ordinary or input variable) appears in 
the selected subgoal of $G_{g_{i+1}}$. 
\end{definition} 

\begin{theorem} 
\label{th-term-dec} 
Let $D$ be a derivation such that for all $r\geq 2$ 
$D$ has a prefix of form (\ref{eq3}), which
has the term-size decrease property.
$D$ contains
an infinite chain of substitutions of form (\ref{ins-chain})
for some input variable $I$ at the root node of $D$.
\end{theorem} 

\begin{proof}
Due to the term-size decrease property of the prefix of $D$ which holds
for all $r\geq 2$, 
$D$ contains an infinite number of substitutions of the form $X/f(...)$,
where $X$ is an input variable.
Assume, on the contrary, that $D$ does not contain such an 
infinite chain of form (\ref{ins-chain}). 
Let $M$ be the longest length of substitutions of
form (\ref{ins-chain}) for each input variable $I$
at the root node of $D$. Note that each input variable can be substituted 
only by a constant or function. For each substitution $X/f(...)$ with $X$ 
an input variable, assume $f(...)$ contains at most $N$ variables
(i.e., it introduces at most $N$ new input variables). 
Then, $D$ contains at most $K*(N^0 + N^1 + ... + N^M)$ substitutions of the form $X/f(...)$,
where $K$ is the number of input variables at the root node of $D$
and $X$ is an input variable. This contradicts the condition that
$D$ contains an infinite number of such substitutions. 
\end{proof}

LP-check and the term-size decrease property approximate conditions (i) and (ii)
of Theorem \ref{th-main}, respectively. So, we can guess an infinite extension (\ref{eq2}) 
from a prefix (\ref{eq3}) by combining the two mechanisms,
as described in the following algorithm.

\begin{algorithm}[{\em Predicting termination of a logic program}] 
\label{alg1}
\begin{tabbing}
Input: \= A logic program $P$, a (concrete or moded) query $Q_0$,
and a repetition number $r\geq 2$ \\
\> ($r = 3$ is recommended).\\
Output: {\em terminating}, {\em predicted-terminating}, or
{\em predicted-non-terminating}. \\
Method: Apply the following procedure.\\
{\em procedure} TPoLP($P, Q_0, r$)\\
\{
\end{tabbing}   
\begin{enumerate} 
\item
Initially, set $L = 0$. 
Construct the moded generalized SLDNF-tree $MT_{Q_0}$ of $P$ for $Q_0$ in the way that whenever a 
prefix $D_x$ of the form
\[N_0:G_0\Rightarrow_{C_0} ...\ N_{g_1}:G_{g_1}\Rightarrow_{C_k} ...\   
N_{g_2}:G_{g_2}\Rightarrow_{C_k} ...\ N_{g_r}:G_{g_r}\Rightarrow_{C_k}\]
is produced which satisfies conditions (a) and (b) of LP-check,
if $D_x$ does not have the term-size decrease property then goto \ref{no}; else set $L = 1$
and extend $D_x$ from the node $N_{g_r}$ with the looping clause $C_k$ skipped.

\item
Return {\em terminating} if $L = 0$; otherwise, return {\em predicted-terminating}.  

\item
\label{no}
Return {\em predicted-non-terminating}. 
\end{enumerate}
\}  
\end{algorithm} 

Starting from the root node $N_0:G_0$, we generate 
derivations of a moded generalized SLDNF-tree $MT_{Q_0}$ 
step by step. If a prefix $D_x$ of form (\ref{eq3})
is generated which satisfies conditions (a) and (b) of LP-check, 
then by Theorem \ref{th1} $D_x$ is very
likely to extend infinitely in $MT_{Q_0}$ (via the looping clause $C_k$). 
By Theorem \ref{th2}, however, the extension of $D_x$ may not have infinite moded instances
in $MF_{Q_0}$. So in this case, we further check if $D_x$ has the term-size decrease property.
If not, by Theorem \ref{th-main} $D_x$ is very likely to 
have moded instances that extend infinitely 
in $MF_{Q_0}$. Algorithm \ref{alg1} then predicts non-terminating for $Q_0$
by returning an answer {\em predicted-non-terminating}.
If $D_x$ has the term-size decrease property, however, 
we continue to extend $D_x$ from $N_{g_r}$ by skipping the clause $C_k$
(i.e., the derivation via $C_k$ is cut at $N_{g_r}$ by LP-check). 

When the answer is not 
{\em predicted-non-terminating}, we distinguish between two cases:
(1) $L = 0$. This shows that no derivation was cut by LP-check 
during the construction of $MT_{Q_0}$. Algorithm \ref{alg1}
concludes terminating for $Q_0$ by 
returning an answer {\em terminating}.
(2) $L = 1$. This means that some derivations were cut by LP-check, all of which
have the term-size decrease property. Algorithm \ref{alg1} then 
predicts terminating for $Q_0$ by
returning an answer {\em predicted-terminating}.

Note that for a concrete query $Q_0$, no derivation has the term-size decrease property.
Therefore, Algorithm \ref{alg1} returns {\em predicted-non-terminating} for $Q_0$ 
once a prefix of a derivation satisfying the conditions of LP-check is generated.

We prove the termination property of Algorithm \ref{alg1}. 

\begin{proposition}
For any logic program $P$, concrete/moded query $Q_0$, and
repetition number $r$, the procedure TPoLP($P, Q_0, r$) terminates.
\end{proposition}

\begin{proof} 
The procedure TPoLP constructs $MT_{Q_0}$
while applying LP-check to cut possible infinite derivations.
Since LP-check is a complete loop check, it cuts all
infinite derivations at some depth. This means that $MT_{Q_0}$ after
cut by LP-check is finite. So, TPoLP($P, Q_0, r$) terminates.
\end{proof}

Algorithm \ref{alg1} yields a heuristic answer, 
{\em predicted-terminating} or {\em predicted-non-terminating},
or an exact answer {\em terminating}, as shown by the 
following theorem.

\begin{theorem}
\label{th-terminating}
A logic program $P$ terminates for a query $Q_0$ if Algorithm \ref{alg1} returns {\em terminating}.
\end{theorem}

\begin{proof}
If Algorithm \ref{alg1} returns {\em terminating}, 
no derivations were cut by LP-check, so
the moded generalized SLDNF-tree $MT_{Q_0}$ for $Q_0$ is finite.
By Corollary \ref{cor-main}, the logic program $P$ terminates for the query $Q_0$. 
\end{proof}

In the following examples, we choose a repetition number $r = 3$.

\begin{example}
\label{alg-eg1}
Consider Figure \ref{fig1}. 
Since the prefix $D_x$ between $N_0$ and $N_4$
satisfies the conditions of LP-check, Algorithm \ref{alg1} 
concludes that the derivation may
extend infinitely in $MT_{Q_0}$. It then
checks the term-size decrease property to see 
if $D_x$ has moded instances that would extend infinitely in $MF_{Q_0}$.
Clearly, $D_x$ has the term-size decrease property. So Algorithm \ref{alg1} 
skips $C_{p_2}$ at $N_4$ (the branch is cut by LP-check).
Consequently, Algorithm \ref{alg1} predicts terminating for $p({\cal I})$ by
returning an answer {\em predicted-terminating}.
This prediction is correct; see Example \ref{eg0-1}.      
\end{example}

\begin{example} 
Consider Figure \ref{fig2}. All the derivations
starting at $N_0$ and ending at $N_2$ satisfy the conditions of LP-check, so 
they are cut at $N_2$. Since the derivations in $MT_{Q_0^1}$ and $MT_{Q_0^3}$
have the term-size decrease property, 
Algorithm \ref{alg1} returns {\em predicted-terminating}
for $Q_0^1$ and $Q_0^3$. Since the derivation in 
$MT_{Q_0^2}$ does not have the term-size decrease property,
Algorithm \ref{alg1} returns {\em predicted-non-terminating} for $Q_0^2$. 
These predictions are all correct; see Example \ref{eg2}.  
\end{example} 
 
\begin{example}  
Consider the following logic program $P_3$:
\begin{tabbing} 
$\qquad\quad mult(s(X),Y,Z)\leftarrow mult(X,Y,U),add(U,Y,Z)$. \`$C_{m_1}$\\ 
$\qquad\quad mult(0,Y,0)$.       \`$C_{m_2}$ \\[.06in]
$\qquad\quad add(s(X),Y,s(Z))\leftarrow add(X,Y,Z)$.   \`$C_{a_1}$\\ 
$\qquad\quad add(0,Y,Y)$.   \`$C_{a_2}$
\end{tabbing}
$MQ(P_3)$ consists of fourteen moded queries, 
seven for predicate $mult$ and seven for predicate $add$. 
Applying Algorithm \ref{alg1} yields the following result: (1)
$P_3$ is {\em predicted-terminating} for all moded queries to 
$add$ except $add(V_1,{\cal I}_2,V_3)$ for which $P_3$ 
is {\em predicted-non-terminating},
and (2) $P_3$ is {\em predicted-terminating} for $mult({\cal I}_1,{\cal I}_2,V_3)$
and $mult({\cal I}_1,{\cal I}_2,$ ${\cal I}_3)$, but is {\em predicted-non-terminating}
for the remaining moded queries to $mult$. 
For illustration, we depict two moded generalized
SLDNF-trees for $mult({\cal I},V_2,V_3)$ and $mult({\cal I}_1,$ ${\cal I}_2,V_3)$,
as shown in Figures \ref{fig-eg3} (a) and (b), respectively.
In the two moded generalized SLDNF-trees, the prefix from $N_0$ down to $N_2$ satisfies 
the conditions of LP-check and has the term-size decrease property, so clause $C_{m_1}$
is skipped when expanding $N_2$. When the derivation is extended to $N_6$,
the conditions of LP-check are satisfied again, where $G_6$ is a loop goal of $G_5$ 
that is a loop goal of $G_4$. 
Since the derivation for $mult({\cal I},V_2,V_3)$ 
(Figure \ref{fig-eg3} (a)) does not have the term-size decrease property, Algorithm \ref{alg1}
returns an answer, {\em predicted-non-terminating}, for this moded query. 
The derivation for $mult({\cal I}_1,{\cal I}_2,V_3)$ 
(Figure \ref{fig-eg3} (b)) has the term-size decrease property, so clause 
$C_{a_1}$ is skipped when expanding $N_6$. For simplicity, we omitted all derivations
leading to a success leaf. Because all derivations satisfying 
the conditions of LP-check have the term-size decrease property, Algorithm \ref{alg1} ends with
an answer, {\em predicted-terminating}, for $mult({\cal I}_1,{\cal I}_2,V_3)$.
It is then immediately inferred by Theorem \ref{th-inc} 
that $P_3$ is {\em predicted-terminating} 
for $mult({\cal I}_1,{\cal I}_2,{\cal I}_3)$.
It is not difficult to verify that all these predictions are correct.
\begin{figure*}[ht]
\begin{tabbing}
$\qquad$
\setlength{\unitlength}{3552sp}%
\begingroup\makeatletter\ifx\SetFigFont\undefined%
\gdef\SetFigFont#1#2#3#4#5{%
  \reset@font\fontsize{#1}{#2pt}%
  \fontfamily{#3}\fontseries{#4}\fontshape{#5}%
  \selectfont}%
\fi\endgroup%
\begin{picture}(5025,3889)(1276,-3716)
\put(1801,-3661){\makebox(0,0)[lb]{\smash{{\SetFigFont{9}{10.8}{\rmdefault}{\mddefault}{\updefault}{\color[rgb]{0,0,0}(a)}%
}}}}
\put(5401,-3676){\makebox(0,0)[lb]{\smash{{\SetFigFont{9}{10.8}{\rmdefault}{\mddefault}{\updefault}{\color[rgb]{0,0,0}(b)}%
}}}}
\thinlines
{\color[rgb]{0,0,0}\put(1651, 14){\vector( 0,-1){300}}
}%
{\color[rgb]{0,0,0}\put(1651,-511){\vector( 0,-1){300}}
}%
{\color[rgb]{0,0,0}\put(1651,-2386){\vector( 0,-1){300}}
}%
{\color[rgb]{0,0,0}\put(1651,-2911){\vector( 0,-1){300}}
}%
{\color[rgb]{0,0,0}\put(1651,-1861){\vector( 0,-1){300}}
}%
{\color[rgb]{0,0,0}\put(1651,-1036){\vector( 0,-1){600}}
}%
{\color[rgb]{0,0,0}\put(5251, -1){\vector( 0,-1){300}}
}%
{\color[rgb]{0,0,0}\put(5251,-526){\vector( 0,-1){300}}
}%
{\color[rgb]{0,0,0}\put(5251,-2401){\vector( 0,-1){300}}
}%
{\color[rgb]{0,0,0}\put(5251,-2926){\vector( 0,-1){300}}
}%
{\color[rgb]{0,0,0}\put(5251,-1876){\vector( 0,-1){300}}
}%
{\color[rgb]{0,0,0}\put(5251,-1039){\vector( 0,-1){612}}
}%
\put(1276, 89){\makebox(0,0)[lb]{\smash{{\SetFigFont{8}{9.6}{\rmdefault}{\mddefault}{\updefault}{\color[rgb]{0,0,0}$N_0$:  $mult(\underline{I}, V_2, V_3)$}%
}}}}
\put(1276,-436){\makebox(0,0)[lb]{\smash{{\SetFigFont{8}{9.6}{\rmdefault}{\mddefault}{\updefault}{\color[rgb]{0,0,0}$N_1$:  $mult(\underline{X_1},V_2,U_1),add(U_1,V_2,V_3)$ }%
}}}}
\put(1276,-961){\makebox(0,0)[lb]{\smash{{\SetFigFont{8}{9.6}{\rmdefault}{\mddefault}{\updefault}{\color[rgb]{0,0,0}$N_2$:  $mult(\underline{X_2},V_2,U_2),add(U_2,V_2,U_1),$}%
}}}}
\put(1351,-136){\makebox(0,0)[lb]{\smash{{\SetFigFont{6}{7.2}{\rmdefault}{\mddefault}{\updefault}{\color[rgb]{0,0,0}$C_{m_1}$}%
}}}}
\put(1351,-661){\makebox(0,0)[lb]{\smash{{\SetFigFont{6}{7.2}{\rmdefault}{\mddefault}{\updefault}{\color[rgb]{0,0,0}$C_{m_1}$}%
}}}}
\put(1726,-136){\makebox(0,0)[lb]{\smash{{\SetFigFont{6}{7.2}{\rmdefault}{\mddefault}{\updefault}{\color[rgb]{0,0,0}$\theta_0 = \{\underline{I}/s(X_1),Y_1/V_2,Z_1/V_3\}$}%
}}}}
\put(1726,-661){\makebox(0,0)[lb]{\smash{{\SetFigFont{6}{7.2}{\rmdefault}{\mddefault}{\updefault}{\color[rgb]{0,0,0}$\theta_1 = \{\underline{X_1}/s(X_2),Y_2/V_2,Z_2/U_1\}$}%
}}}}
\put(1276,-2836){\makebox(0,0)[lb]{\smash{{\SetFigFont{8}{9.6}{\rmdefault}{\mddefault}{\updefault}{\color[rgb]{0,0,0}$N_5$:  $add(X_3,s(X_3),Z_3)$ }%
}}}}
\put(1276,-3361){\makebox(0,0)[lb]{\smash{{\SetFigFont{8}{9.6}{\rmdefault}{\mddefault}{\updefault}{\color[rgb]{0,0,0}$N_6$:  $add(X_4,s(s(X_4)),Z_4)$ }%
}}}}
\put(1351,-2536){\makebox(0,0)[lb]{\smash{{\SetFigFont{6}{7.2}{\rmdefault}{\mddefault}{\updefault}{\color[rgb]{0,0,0}$C_{a_1}$}%
}}}}
\put(1351,-3061){\makebox(0,0)[lb]{\smash{{\SetFigFont{6}{7.2}{\rmdefault}{\mddefault}{\updefault}{\color[rgb]{0,0,0}$C_{a_1}$}%
}}}}
\put(1726,-2536){\makebox(0,0)[lb]{\smash{{\SetFigFont{6}{7.2}{\rmdefault}{\mddefault}{\updefault}{\color[rgb]{0,0,0}$\theta_4 = \{V_2/s(X_3),V_3/s(Z_3)\}$}%
}}}}
\put(1726,-3061){\makebox(0,0)[lb]{\smash{{\SetFigFont{6}{7.2}{\rmdefault}{\mddefault}{\updefault}{\color[rgb]{0,0,0}$\theta_5 = \{X_3/s(X_4),Z_3/s(Z_4)\}$}%
}}}}
\put(1276,-1786){\makebox(0,0)[lb]{\smash{{\SetFigFont{8}{9.6}{\rmdefault}{\mddefault}{\updefault}{\color[rgb]{0,0,0}$N_3$:  $add(0,V_2,U_1),add(U_1,V_2,V_3)$ }%
}}}}
\put(1276,-2311){\makebox(0,0)[lb]{\smash{{\SetFigFont{8}{9.6}{\rmdefault}{\mddefault}{\updefault}{\color[rgb]{0,0,0}$N_4$:  $add(V_2,V_2,V_3)$ }%
}}}}
\put(1351,-1486){\makebox(0,0)[lb]{\smash{{\SetFigFont{6}{7.2}{\rmdefault}{\mddefault}{\updefault}{\color[rgb]{0,0,0}$C_{m_2}$}%
}}}}
\put(1351,-2011){\makebox(0,0)[lb]{\smash{{\SetFigFont{6}{7.2}{\rmdefault}{\mddefault}{\updefault}{\color[rgb]{0,0,0}$C_{a_2}$}%
}}}}
\put(1726,-1486){\makebox(0,0)[lb]{\smash{{\SetFigFont{6}{7.2}{\rmdefault}{\mddefault}{\updefault}{\color[rgb]{0,0,0}$\theta_2 = \{\underline{X_2}/0,Y_3/V_2,U_2/0\}$}%
}}}}
\put(1726,-2011){\makebox(0,0)[lb]{\smash{{\SetFigFont{6}{7.2}{\rmdefault}{\mddefault}{\updefault}{\color[rgb]{0,0,0}$\theta_3 = \{Y_4/V_2,U_1/V_2\}$}%
}}}}
\put(2776,-1186){\makebox(0,0)[lb]{\smash{{\SetFigFont{8}{9.6}{\rmdefault}{\mddefault}{\updefault}{\color[rgb]{0,0,0}$add(U_1,V_2,V_3)$ }%
}}}}
\put(4876, 74){\makebox(0,0)[lb]{\smash{{\SetFigFont{8}{9.6}{\rmdefault}{\mddefault}{\updefault}{\color[rgb]{0,0,0}$N_0$:  $mult(\underline{I_1}, \underline{I_2}, V_3)$}%
}}}}
\put(4876,-451){\makebox(0,0)[lb]{\smash{{\SetFigFont{8}{9.6}{\rmdefault}{\mddefault}{\updefault}{\color[rgb]{0,0,0}$N_1$:  $mult(\underline{X_1},\underline{I_2},U_1),add(U_1,\underline{I_2},V_3)$ }%
}}}}
\put(4876,-976){\makebox(0,0)[lb]{\smash{{\SetFigFont{8}{9.6}{\rmdefault}{\mddefault}{\updefault}{\color[rgb]{0,0,0}$N_2$:  $mult(\underline{X_2},\underline{I_2},U_2),add(U_2,\underline{I_2},U_1),$}%
}}}}
\put(4951,-151){\makebox(0,0)[lb]{\smash{{\SetFigFont{6}{7.2}{\rmdefault}{\mddefault}{\updefault}{\color[rgb]{0,0,0}$C_{m_1}$}%
}}}}
\put(4951,-676){\makebox(0,0)[lb]{\smash{{\SetFigFont{6}{7.2}{\rmdefault}{\mddefault}{\updefault}{\color[rgb]{0,0,0}$C_{m_1}$}%
}}}}
\put(5326,-151){\makebox(0,0)[lb]{\smash{{\SetFigFont{6}{7.2}{\rmdefault}{\mddefault}{\updefault}{\color[rgb]{0,0,0}$\theta_0 = \{\underline{I_1}/s(X_1),Y_1/\underline{I_2},Z_1/V_3\}$}%
}}}}
\put(5326,-676){\makebox(0,0)[lb]{\smash{{\SetFigFont{6}{7.2}{\rmdefault}{\mddefault}{\updefault}{\color[rgb]{0,0,0}$\theta_1 = \{\underline{X_1}/s(X_2),Y_2/\underline{I_2},Z_2/U_1\}$}%
}}}}
\put(4876,-3376){\makebox(0,0)[lb]{\smash{{\SetFigFont{8}{9.6}{\rmdefault}{\mddefault}{\updefault}{\color[rgb]{0,0,0}$N_6$:  $add(\underline{X_4},s(s(\underline{X_4})),Z_4)$ }%
}}}}
\put(4951,-2551){\makebox(0,0)[lb]{\smash{{\SetFigFont{6}{7.2}{\rmdefault}{\mddefault}{\updefault}{\color[rgb]{0,0,0}$C_{a_1}$}%
}}}}
\put(4951,-3076){\makebox(0,0)[lb]{\smash{{\SetFigFont{6}{7.2}{\rmdefault}{\mddefault}{\updefault}{\color[rgb]{0,0,0}$C_{a_1}$}%
}}}}
\put(5326,-2551){\makebox(0,0)[lb]{\smash{{\SetFigFont{6}{7.2}{\rmdefault}{\mddefault}{\updefault}{\color[rgb]{0,0,0}$\theta_4 = \{\underline{I_2}/s(X_3),V_3/s(Z_3)\}$}%
}}}}
\put(5326,-3076){\makebox(0,0)[lb]{\smash{{\SetFigFont{6}{7.2}{\rmdefault}{\mddefault}{\updefault}{\color[rgb]{0,0,0}$\theta_5 = \{\underline{X_3}/s(X_4),Z_3/s(Z_4)\}$}%
}}}}
\put(4876,-1801){\makebox(0,0)[lb]{\smash{{\SetFigFont{8}{9.6}{\rmdefault}{\mddefault}{\updefault}{\color[rgb]{0,0,0}$N_3$:  $add(0,\underline{I_2},U_1),add(U_1,\underline{I_2},V_3)$ }%
}}}}
\put(4876,-2326){\makebox(0,0)[lb]{\smash{{\SetFigFont{8}{9.6}{\rmdefault}{\mddefault}{\updefault}{\color[rgb]{0,0,0}$N_4$:  $add(\underline{I_2},\underline{I_2},V_3)$ }%
}}}}
\put(4951,-1501){\makebox(0,0)[lb]{\smash{{\SetFigFont{6}{7.2}{\rmdefault}{\mddefault}{\updefault}{\color[rgb]{0,0,0}$C_{m_2}$}%
}}}}
\put(4951,-2026){\makebox(0,0)[lb]{\smash{{\SetFigFont{6}{7.2}{\rmdefault}{\mddefault}{\updefault}{\color[rgb]{0,0,0}$C_{a_2}$}%
}}}}
\put(5326,-1501){\makebox(0,0)[lb]{\smash{{\SetFigFont{6}{7.2}{\rmdefault}{\mddefault}{\updefault}{\color[rgb]{0,0,0}$\theta_2 = \{\underline{X_2}/0,Y_3/\underline{I_2},U_2/0\}$}%
}}}}
\put(5326,-2026){\makebox(0,0)[lb]{\smash{{\SetFigFont{6}{7.2}{\rmdefault}{\mddefault}{\updefault}{\color[rgb]{0,0,0}$\theta_3 = \{Y_4/\underline{I_2},U_1/\underline{I_2}\}$}%
}}}}
\put(4876,-2851){\makebox(0,0)[lb]{\smash{{\SetFigFont{8}{9.6}{\rmdefault}{\mddefault}{\updefault}{\color[rgb]{0,0,0}$N_5$:  $add(\underline{X_3},s(\underline{X_3}),Z_3)$ }%
}}}}
\put(6301,-1186){\makebox(0,0)[lb]{\smash{{\SetFigFont{8}{9.6}{\rmdefault}{\mddefault}{\updefault}{\color[rgb]{0,0,0}$add(U_1,\underline{I_2},V_3)$}%
}}}}
\end{picture}%
\end{tabbing}
\caption{Two moded generalized SLDNF-trees of $P_3$ generated by Algorithm \ref{alg1}.}\label{fig-eg3} 
\end{figure*}  
\end{example} 

AProVE07 \cite{GST06}, NTI \cite{PM06,Pay06}, Polytool \cite{ND05,NDGS06} and TALP \cite{OCM00}
are four well-known state-of-the-art analyzers. 
NTI proves non-termination, while the others prove termination.
The Termination Competition 2007 
(http://www.lri.fr/\verb+~+marche/termination-competition/2007) 
reports their latest performance.
We borrow three representative logic programs
from the competition website to further demonstrate the effectiveness of our termination
prediction.

\begin{example}
\label{subset1} 
Consider the following logic program coming from 
the Termination Competition 2007 
with Problem id {\em LP/talp/apt - subset1} and 
difficulty rating $100\%$. 
AProVE07, NTI, Polytool and TALP all failed to prove/disprove its termination by
yielding an answer ``don't know'' in the competition.
\begin{tabbing} 
$\qquad$ \= $P_4: \quad$ \= $member1(X,[Y|Xs]) \leftarrow member1(X,Xs).$ \`$C_{m_1}$\\ 
\>\> $member1(X,[X|Xs]).$ \`$C_{m_2}$ \\[.06in] 
\>\> $subset1([X|Xs],Ys) :- member1(X,Ys), subset1(Xs,Ys).$ \`$C_{s_1}$\\ 
\>\> $subset1([],Ys).$ \`$C_{s_2}$ \\[.06in] 
\> Query Mode:  $\ subset1(o,i)$. 
\end{tabbing}
The query mode $subset1(o,i)$ means that the second argument
of any query must be a ground term, while the first one
can be an arbitrary term. Then, to prove the termination property
of $P_4$ with this query mode is to prove the termination for the moded query
$Q_0 = subset1(V, {\cal I})$. Applying Algorithm \ref{alg1}
generates a moded generalized SLDNF-tree
as shown in Figure \ref{fig-subset1}. The prefix from $N_0$ down to $N_3$ satisfies 
the conditions of LP-check and has the term-size decrease property, so clause $C_{m_1}$
is skipped when expanding $N_3$. When the derivation is extended to $N_{10}$,
the conditions of LP-check are satisfied again, where $G_{10}$ is a loop goal of $G_9$ 
that is a loop goal of $G_8$. Since the derivation has the term-size decrease property,
$N_{10}$ is expanded by $C_{m_2}$. 
\begin{figure*}[htb]
\begin{tabbing}
\setlength{\unitlength}{3276sp}%
\begingroup\makeatletter\ifx\SetFigFont\undefined%
\gdef\SetFigFont#1#2#3#4#5{%
  \reset@font\fontsize{#1}{#2pt}%
  \fontfamily{#3}\fontseries{#4}\fontshape{#5}%
  \selectfont}%
\fi\endgroup%
\begin{picture}(6521,6415)(751,-6092)
\thinlines
{\color[rgb]{0,0,0}\put(3376,-4036){\line( 0,-1){225}}
\put(3376,-4261){\line( 1, 0){3750}}
}%
{\color[rgb]{0,0,0}\put(3376,-3511){\line( 0,-1){225}}
\put(3376,-3736){\line( 1, 0){3750}}
}%
{\color[rgb]{0,0,0}\put(3376,-4561){\line( 0,-1){264}}
\put(3376,-4825){\line( 1, 0){2625}}
\put(6001,-4825){\line( 0,-1){261}}
}%
{\color[rgb]{0,0,0}\put(1126,-1936){\vector( 0,-1){300}}
}%
{\color[rgb]{0,0,0}\put(1126,-2461){\vector( 0,-1){300}}
}%
{\color[rgb]{0,0,0}\put(1126,-886){\vector( 0,-1){300}}
}%
{\color[rgb]{0,0,0}\put(1126,-1411){\vector( 0,-1){300}}
}%
{\color[rgb]{0,0,0}\put(1126,164){\vector( 0,-1){300}}
}%
{\color[rgb]{0,0,0}\put(1126,-361){\vector( 0,-1){300}}
}%
{\color[rgb]{0,0,0}\put(1126,-2986){\vector( 0,-1){300}}
}%
{\color[rgb]{0,0,0}\put(1126,-3511){\vector( 0,-1){300}}
}%
{\color[rgb]{0,0,0}\put(1126,-4036){\vector( 0,-1){300}}
}%
{\color[rgb]{0,0,0}\put(1126,-4561){\vector( 0,-1){300}}
}%
{\color[rgb]{0,0,0}\put(7126,-4261){\vector( 0,-1){300}}
}%
{\color[rgb]{0,0,0}\put(7126,-4786){\vector( 0,-1){300}}
}%
{\color[rgb]{0,0,0}\put(1126,-5086){\vector( 0,-1){300}}
}%
{\color[rgb]{0,0,0}\put(1126,-5611){\vector( 0,-1){300}}
}%
{\color[rgb]{0,0,0}\put(7126,-3736){\vector( 0,-1){232}}
}%
{\color[rgb]{0,0,0}\put(6001,-5086){\vector( 0,-1){300}}
}%
{\color[rgb]{0,0,0}\put(6001,-5611){\vector( 0,-1){300}}
}%
\put(751,239){\makebox(0,0)[lb]{\smash{{\SetFigFont{7}{8.4}{\rmdefault}{\mddefault}{\updefault}{\color[rgb]{0,0,0}$N_0$:  $subset1(V, \underline{I})$}%
}}}}
\put(826,-2086){\makebox(0,0)[lb]{\smash{{\SetFigFont{6}{7.2}{\rmdefault}{\mddefault}{\updefault}{\color[rgb]{0,0,0}$C_{s_1}$}%
}}}}
\put(826,-2611){\makebox(0,0)[lb]{\smash{{\SetFigFont{6}{7.2}{\rmdefault}{\mddefault}{\updefault}{\color[rgb]{0,0,0}$C_{m_1}$}%
}}}}
\put(826,-1036){\makebox(0,0)[lb]{\smash{{\SetFigFont{6}{7.2}{\rmdefault}{\mddefault}{\updefault}{\color[rgb]{0,0,0}$C_{m_1}$}%
}}}}
\put(826,-1561){\makebox(0,0)[lb]{\smash{{\SetFigFont{6}{7.2}{\rmdefault}{\mddefault}{\updefault}{\color[rgb]{0,0,0}$C_{m_2}$}%
}}}}
\put(1201,-1036){\makebox(0,0)[lb]{\smash{{\SetFigFont{6}{7.2}{\rmdefault}{\mddefault}{\updefault}{\color[rgb]{0,0,0}$\theta_2 = \{\underline{Xs_1}/[Y_2|Xs_2]\}$}%
}}}}
\put(1201,-1561){\makebox(0,0)[lb]{\smash{{\SetFigFont{6}{7.2}{\rmdefault}{\mddefault}{\updefault}{\color[rgb]{0,0,0}$\theta_3 = \{\underline{Xs_2}/[X|Xs_3]\}$}%
}}}}
\put(751,-286){\makebox(0,0)[lb]{\smash{{\SetFigFont{7}{8.4}{\rmdefault}{\mddefault}{\updefault}{\color[rgb]{0,0,0}$N_1$:  $member1(X,\underline{I}),subset1(Xs,\underline{I})$ }%
}}}}
\put(826, 14){\makebox(0,0)[lb]{\smash{{\SetFigFont{6}{7.2}{\rmdefault}{\mddefault}{\updefault}{\color[rgb]{0,0,0}$C_{s_1}$}%
}}}}
\put(826,-511){\makebox(0,0)[lb]{\smash{{\SetFigFont{6}{7.2}{\rmdefault}{\mddefault}{\updefault}{\color[rgb]{0,0,0}$C_{m_1}$}%
}}}}
\put(1201,-511){\makebox(0,0)[lb]{\smash{{\SetFigFont{6}{7.2}{\rmdefault}{\mddefault}{\updefault}{\color[rgb]{0,0,0}$\theta_1 = \{\underline{I}/[Y_1|Xs_1]\}$}%
}}}}
\put(751,-811){\makebox(0,0)[lb]{\smash{{\SetFigFont{7}{8.4}{\rmdefault}{\mddefault}{\updefault}{\color[rgb]{0,0,0}$N_2$:  $member1(X,\underline{Xs_1}),subset1(Xs,[\underline{Y_1}|\underline{Xs_1]})$ }%
}}}}
\put(751,-2386){\makebox(0,0)[lb]{\smash{{\SetFigFont{7}{8.4}{\rmdefault}{\mddefault}{\updefault}{\color[rgb]{0,0,0}$N_5$:  $member1(X_1,[\underline{Y_1}|[\underline{Y_2}|[\underline{X}|\underline{Xs_3}]]]),subset1(Xs_4,[\underline{Y_1}|[\underline{Y_2}|[\underline{X}|\underline{Xs_3}]]])$ }%
}}}}
\put(751,-2911){\makebox(0,0)[lb]{\smash{{\SetFigFont{7}{8.4}{\rmdefault}{\mddefault}{\updefault}{\color[rgb]{0,0,0}$N_6$:  $member1(X_1,[\underline{Y_2}|[\underline{X}|\underline{Xs_3}]]),subset1(Xs_4,[\underline{Y_1}|[\underline{Y_2}|[\underline{X}|\underline{Xs_3}]]])$ }%
}}}}
\put(751,-3961){\makebox(0,0)[lb]{\smash{{\SetFigFont{7}{8.4}{\rmdefault}{\mddefault}{\updefault}{\color[rgb]{0,0,0}$N_8$:  $member1(X_1,\underline{Xs_3}),subset1(Xs_4,[\underline{Y_1}|[\underline{Y_2}|[\underline{X}|\underline{Xs_3}]]])$ }%
}}}}
\put(751,-4486){\makebox(0,0)[lb]{\smash{{\SetFigFont{7}{8.4}{\rmdefault}{\mddefault}{\updefault}{\color[rgb]{0,0,0}$N_9$:  $member1(X_1,\underline{Xs_5}),subset1(Xs_4,[\underline{Y_1}|[\underline{Y_2}|[\underline{X}|[\underline{Y_3}|\underline{Xs_5}]]]])$ }%
}}}}
\put(751,-1861){\makebox(0,0)[lb]{\smash{{\SetFigFont{7}{8.4}{\rmdefault}{\mddefault}{\updefault}{\color[rgb]{0,0,0}$N_4$:  $subset1(Xs,[\underline{Y_1}|[\underline{Y_2}|[\underline{X}|\underline{Xs_3}]]])$ }%
}}}}
\put(751,-1336){\makebox(0,0)[lb]{\smash{{\SetFigFont{7}{8.4}{\rmdefault}{\mddefault}{\updefault}{\color[rgb]{0,0,0}$N_3$:  $member1(X,\underline{Xs_2}),subset1(Xs,[\underline{Y_1}|[\underline{Y_2}|\underline{Xs_2}]])$ }%
}}}}
\put(751,-3436){\makebox(0,0)[lb]{\smash{{\SetFigFont{7}{8.4}{\rmdefault}{\mddefault}{\updefault}{\color[rgb]{0,0,0}$N_7$:  $member1(X_1,[\underline{X}|\underline{Xs_3}]),subset1(Xs_4,[\underline{Y_1}|[\underline{Y_2}|[\underline{X}|\underline{Xs_3}]]])$ }%
}}}}
\put(1201, 14){\makebox(0,0)[lb]{\smash{{\SetFigFont{6}{7.2}{\rmdefault}{\mddefault}{\updefault}{\color[rgb]{0,0,0}$\theta_0 = \{V/[X|X_s]\}$}%
}}}}
\put(826,-3136){\makebox(0,0)[lb]{\smash{{\SetFigFont{6}{7.2}{\rmdefault}{\mddefault}{\updefault}{\color[rgb]{0,0,0}$C_{m_1}$}%
}}}}
\put(826,-3661){\makebox(0,0)[lb]{\smash{{\SetFigFont{6}{7.2}{\rmdefault}{\mddefault}{\updefault}{\color[rgb]{0,0,0}$C_{m_1}$}%
}}}}
\put(826,-4186){\makebox(0,0)[lb]{\smash{{\SetFigFont{6}{7.2}{\rmdefault}{\mddefault}{\updefault}{\color[rgb]{0,0,0}$C_{m_1}$}%
}}}}
\put(826,-4711){\makebox(0,0)[lb]{\smash{{\SetFigFont{6}{7.2}{\rmdefault}{\mddefault}{\updefault}{\color[rgb]{0,0,0}$C_{m_1}$}%
}}}}
\put(1201,-2086){\makebox(0,0)[lb]{\smash{{\SetFigFont{6}{7.2}{\rmdefault}{\mddefault}{\updefault}{\color[rgb]{0,0,0}$\theta_4 = \{Xs/[X_1|Xs_4]\}$}%
}}}}
\put(1201,-4186){\makebox(0,0)[lb]{\smash{{\SetFigFont{6}{7.2}{\rmdefault}{\mddefault}{\updefault}{\color[rgb]{0,0,0}$\theta_8 = \{\underline{Xs_3}/[Y_3|Xs_5]\}$}%
}}}}
\put(1201,-4711){\makebox(0,0)[lb]{\smash{{\SetFigFont{6}{7.2}{\rmdefault}{\mddefault}{\updefault}{\color[rgb]{0,0,0}$\theta_9 = \{\underline{Xs_5}/[Y_4|Xs_6]\}$}%
}}}}
\put(5626,-6061){\makebox(0,0)[lb]{\smash{{\SetFigFont{7}{8.4}{\rmdefault}{\mddefault}{\updefault}{\color[rgb]{0,0,0}$N_{14}$:  $\Box_t$}%
}}}}
\put(3076,-4711){\makebox(0,0)[lb]{\smash{{\SetFigFont{6}{7.2}{\rmdefault}{\mddefault}{\updefault}{\color[rgb]{0,0,0}$C_{m_2}$}%
}}}}
\put(6751,-5236){\makebox(0,0)[lb]{\smash{{\SetFigFont{7}{8.4}{\rmdefault}{\mddefault}{\updefault}{\color[rgb]{0,0,0}$N_{16}$:  $\Box_t$}%
}}}}
\put(6826,-4936){\makebox(0,0)[lb]{\smash{{\SetFigFont{6}{7.2}{\rmdefault}{\mddefault}{\updefault}{\color[rgb]{0,0,0}$C_{s_2}$}%
}}}}
\put(3451,-4711){\makebox(0,0)[lb]{\smash{{\SetFigFont{6}{7.2}{\rmdefault}{\mddefault}{\updefault}{\color[rgb]{0,0,0}$\theta_{9'} = \{\underline{Xs_5}/[X_1|Xs_6]\}$}%
}}}}
\put(3451,-4186){\makebox(0,0)[lb]{\smash{{\SetFigFont{6}{7.2}{\rmdefault}{\mddefault}{\updefault}{\color[rgb]{0,0,0}$\theta_{8'} = \{\underline{Xs_3}/[X_1|Xs_5]\}$}%
}}}}
\put(751,-6061){\makebox(0,0)[lb]{\smash{{\SetFigFont{7}{8.4}{\rmdefault}{\mddefault}{\updefault}{\color[rgb]{0,0,0}$N_{12}$:  $\Box_t$}%
}}}}
\put(751,-5536){\makebox(0,0)[lb]{\smash{{\SetFigFont{7}{8.4}{\rmdefault}{\mddefault}{\updefault}{\color[rgb]{0,0,0}$N_{11}$:  $subset1(Xs_4,[\underline{Y_1}|[\underline{Y_2}|[\underline{X}|[\underline{Y_3}|[\underline{Y_4}|[\underline{X_1}|\underline{Xs_7}]]]]]])$ }%
}}}}
\put(751,-5011){\makebox(0,0)[lb]{\smash{{\SetFigFont{7}{8.4}{\rmdefault}{\mddefault}{\updefault}{\color[rgb]{0,0,0}$N_{10}$:  $member1(X_1,\underline{Xs_6}),subset1(Xs_4,[\underline{Y_1}|[\underline{Y_2}|[\underline{X}|[\underline{Y_3}|[\underline{Y_4}|\underline{Xs_6}]]]]])$ }%
}}}}
\put(826,-5236){\makebox(0,0)[lb]{\smash{{\SetFigFont{6}{7.2}{\rmdefault}{\mddefault}{\updefault}{\color[rgb]{0,0,0}$C_{m_2}$}%
}}}}
\put(826,-5761){\makebox(0,0)[lb]{\smash{{\SetFigFont{6}{7.2}{\rmdefault}{\mddefault}{\updefault}{\color[rgb]{0,0,0}$C_{s_2}$}%
}}}}
\put(1201,-5236){\makebox(0,0)[lb]{\smash{{\SetFigFont{6}{7.2}{\rmdefault}{\mddefault}{\updefault}{\color[rgb]{0,0,0}$\theta_{10} = \{\underline{Xs_6}/[X_1|Xs_7]\}$}%
}}}}
\put(3076,-4186){\makebox(0,0)[lb]{\smash{{\SetFigFont{6}{7.2}{\rmdefault}{\mddefault}{\updefault}{\color[rgb]{0,0,0}$C_{m_2}$}%
}}}}
\put(3451,-3661){\makebox(0,0)[lb]{\smash{{\SetFigFont{6}{7.2}{\rmdefault}{\mddefault}{\updefault}{\color[rgb]{0,0,0}$\theta_{7'} = \{X_1/\underline{X}\}$}%
}}}}
\put(3076,-3661){\makebox(0,0)[lb]{\smash{{\SetFigFont{6}{7.2}{\rmdefault}{\mddefault}{\updefault}{\color[rgb]{0,0,0}$C_{m_2}$}%
}}}}
\put(5701,-5761){\makebox(0,0)[lb]{\smash{{\SetFigFont{6}{7.2}{\rmdefault}{\mddefault}{\updefault}{\color[rgb]{0,0,0}$C_{s_2}$}%
}}}}
\put(4576,-5536){\makebox(0,0)[lb]{\smash{{\SetFigFont{7}{8.4}{\rmdefault}{\mddefault}{\updefault}{\color[rgb]{0,0,0}$N_{13}$:  $subset1(Xs_4,[\underline{Y_1}|[\underline{Y_2}|[\underline{X}|[\underline{Y_3}|[\underline{X_1}|\underline{Xs_6}]]]]])$ }%
}}}}
\put(5026,-4711){\makebox(0,0)[lb]{\smash{{\SetFigFont{7}{8.4}{\rmdefault}{\mddefault}{\updefault}{\color[rgb]{0,0,0}$N_{15}$:  $subset1(Xs_4,[\underline{Y_1}|[\underline{Y_2}|[\underline{X}|[\underline{X_1}|\underline{Xs_5}]]]])$ }%
}}}}
\put(5326,-4111){\makebox(0,0)[lb]{\smash{{\SetFigFont{7}{8.4}{\rmdefault}{\mddefault}{\updefault}{\color[rgb]{0,0,0}$N_{17}$:  $subset1(Xs_4,[\underline{Y_1}|[\underline{Y_2}|[\underline{X}|\underline{Xs_3}]]])$ }%
}}}}
\end{picture}%
\end{tabbing}
\caption{The moded generalized SLDNF-tree of $P_4$ generated by Algorithm \ref{alg1}.}\label{fig-subset1} 
\end{figure*}  

At $N_{11}$ (resp. $N_{13}$ and $N_{15}$), 
the derivation satisfies the conditions of LP-check and has the term-size decrease property,
where $G_{11}$ (resp. $N_{13}$ and $N_{15}$) is a loop goal of $G_4$ 
that is a loop goal of $G_0$. Therefore, $N_{11}$ (resp. $N_{13}$ and $N_{15}$)
is expanded by $C_{s_2}$. When the derivation is extended to $N_{17}$, 
the conditions of LP-check are satisfied, where $G_{17}$ is a loop goal of $G_4$ 
that is a loop goal of $G_0$, but the term-size decrease condition is violated. 
Algorithm \ref{alg1} stops immediately with
an answer, {\em predicted-non-terminating}, for the query   
$Q_0$. It is easy to verify that this prediction is correct.
\end{example} 

\begin{example}
\label{incomplete}  
Consider another logic program in
the Termination Competition 2007  
with Problem id {\em LP/SGST06 - incomplete} and 
difficulty rating $75\%$. Polytool succeeded to prove its termination, 
while AProVE07, NTI and TALP failed.
\begin{tabbing} 
$\qquad$ \= $P_5: \quad$ \= $p(X) \leftarrow  q(f(Y)), p(Y).$ \`$C_{p_1}$\\ 
\>\> $p(g(X)) \leftarrow  p(X).$ \`$C_{p_2}$ \\ 
\>\> $q(g(Y)).$ \`$C_{q_1}$ \\[.06in] 
\> Query Mode:  $\ p(i)$. 
\end{tabbing}
To prove the termination property
of $P_5$ with this query mode is to prove the termination for the moded query
$Q_0 = p({\cal I})$. Applying Algorithm \ref{alg1}
generates a moded generalized SLDNF-tree
as shown in Figure \ref{fig-incomplete}. The prefix from $N_0$ down to $N_4$ satisfies 
the conditions of LP-check and has the term-size decrease property, so clause $C_{p_2}$
is skipped when expanding $N_4$. Algorithm \ref{alg1} yields
an answer {\em predicted-terminating} for the query   
$Q_0$. This prediction is correct.
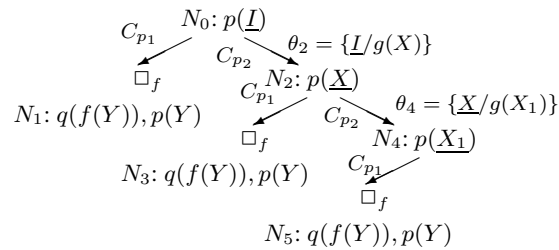
\begin{figure*}[htb]
\begin{center}
\setlength{\unitlength}{3947sp}%
\begingroup\makeatletter\ifx\SetFigFont\undefined%
\gdef\SetFigFont#1#2#3#4#5{%
  \reset@font\fontsize{#1}{#2pt}%
  \fontfamily{#3}\fontseries{#4}\fontshape{#5}%
  \selectfont}%
\fi\endgroup%
\begin{picture}(2821,1482)(1801,-847)
\thinlines
{\color[rgb]{0,0,0}\put(2921,464){\vector(-2,-1){336}}
}%
{\color[rgb]{0,0,0}\put(3256,464){\vector( 2,-1){336}}
}%
{\color[rgb]{0,0,0}\put(3646, 89){\vector(-2,-1){336}}
}%
{\color[rgb]{0,0,0}\put(3856, 89){\vector( 2,-1){336}}
}%
{\color[rgb]{0,0,0}\put(4351,-286){\vector(-2,-1){336}}
}%
\put(2851,539){\makebox(0,0)[lb]{\smash{{\SetFigFont{9}{10.8}{\rmdefault}{\mddefault}{\updefault}{\color[rgb]{0,0,0}$N_0$:  $p(\underline{I})$}%
}}}}
\put(2476,464){\makebox(0,0)[lb]{\smash{{\SetFigFont{8}{9.6}{\rmdefault}{\mddefault}{\updefault}{\color[rgb]{0,0,0}$C_{p_1}$}%
}}}}
\put(2551,164){\makebox(0,0)[lb]{\smash{{\SetFigFont{9}{10.8}{\rmdefault}{\mddefault}{\updefault}{\color[rgb]{0,0,0}$\Box_f$ }%
}}}}
\put(3526,389){\makebox(0,0)[lb]{\smash{{\SetFigFont{8}{9.6}{\rmdefault}{\mddefault}{\updefault}{\color[rgb]{0,0,0}$\theta_2 = \{\underline{I}/g(X)\}$}%
}}}}
\put(3226,-211){\makebox(0,0)[lb]{\smash{{\SetFigFont{9}{10.8}{\rmdefault}{\mddefault}{\updefault}{\color[rgb]{0,0,0}$\Box_f$ }%
}}}}
\put(3901,-361){\makebox(0,0)[lb]{\smash{{\SetFigFont{8}{9.6}{\rmdefault}{\mddefault}{\updefault}{\color[rgb]{0,0,0}$C_{p_1}$}%
}}}}
\put(3976,-586){\makebox(0,0)[lb]{\smash{{\SetFigFont{9}{10.8}{\rmdefault}{\mddefault}{\updefault}{\color[rgb]{0,0,0}$\Box_f$ }%
}}}}
\put(4051,-211){\makebox(0,0)[lb]{\smash{{\SetFigFont{9}{10.8}{\rmdefault}{\mddefault}{\updefault}{\color[rgb]{0,0,0}$N_4$:  $p(\underline{X_1})$ }%
}}}}
\put(3076,314){\makebox(0,0)[lb]{\smash{{\SetFigFont{8}{9.6}{\rmdefault}{\mddefault}{\updefault}{\color[rgb]{0,0,0}$C_{p_2}$}%
}}}}
\put(3376,164){\makebox(0,0)[lb]{\smash{{\SetFigFont{9}{10.8}{\rmdefault}{\mddefault}{\updefault}{\color[rgb]{0,0,0}$N_2$:  $p(\underline{X})$ }%
}}}}
\put(3751,-61){\makebox(0,0)[lb]{\smash{{\SetFigFont{8}{9.6}{\rmdefault}{\mddefault}{\updefault}{\color[rgb]{0,0,0}$C_{p_2}$}%
}}}}
\put(4201, 14){\makebox(0,0)[lb]{\smash{{\SetFigFont{8}{9.6}{\rmdefault}{\mddefault}{\updefault}{\color[rgb]{0,0,0}$\theta_4 = \{\underline{X}/g(X_1)\}$}%
}}}}
\put(3376,-811){\makebox(0,0)[lb]{\smash{{\SetFigFont{9}{10.8}{\rmdefault}{\mddefault}{\updefault}{\color[rgb]{0,0,0}$N_5$:  $q(f(Y)),p(Y)$ }%
}}}}
\put(2476,-436){\makebox(0,0)[lb]{\smash{{\SetFigFont{9}{10.8}{\rmdefault}{\mddefault}{\updefault}{\color[rgb]{0,0,0}$N_3$:  $q(f(Y)),p(Y)$ }%
}}}}
\put(1801,-61){\makebox(0,0)[lb]{\smash{{\SetFigFont{9}{10.8}{\rmdefault}{\mddefault}{\updefault}{\color[rgb]{0,0,0}$N_1$:  $q(f(Y)),p(Y)$ }%
}}}}
\put(3226, 89){\makebox(0,0)[lb]{\smash{{\SetFigFont{8}{9.6}{\rmdefault}{\mddefault}{\updefault}{\color[rgb]{0,0,0}$C_{p_1}$}%
}}}}
\end{picture}%
\end{center}
\caption{The moded generalized SLDNF-tree of $P_5$ generated by Algorithm \ref{alg1}.}\label{fig-incomplete} 
\end{figure*}  
\end{example} 

\begin{example}
\label{incomplete2} 
Consider a third logic program from 
the Termination Competition 2007  
with Problem id {\em LP/SGST06 - incomplete2} and 
difficulty rating $75\%$. In contrast to Example \ref{incomplete},
for this program AProVE07 succeeded to prove its termination, 
while Polytool, NTI and TALP failed.
\begin{tabbing} 
$\qquad$ \= $P_6: \quad$ \= $f(X) \leftarrow  g(s(s(s(X)))).$ \`$C_{f_1}$\\ 
\>\> $f(s(X)) \leftarrow  f(X).$ \`$C_{f_2}$ \\ 
\>\> $g(s(s(s(s(X))))) \leftarrow f(X).$ \`$C_{g_1}$ \\[.06in] 
\> Query Mode:  $\ f(i)$. 
\end{tabbing}
To prove the termination property
of $P_6$ with this query mode is to prove the termination for the moded query
$Q_0 = f({\cal I})$. Applying Algorithm \ref{alg1}
generates a moded generalized SLDNF-tree
as shown in Figure \ref{fig-incomplete2}. 
$C_{f_1}$ and/or $C_{f_2}$ is skipped at $N_4, N_5, N_6, N_9, N_{10},$ $N_{11}, N_{13}, N_{18},$ 
$N_{19}, N_{20}, N_{22}, N_{23}, N_{25}$ and $N_{27}$,
due to the occurrence of the following prefixes which satisfy 
both the conditions of LP-check and the term-size decrease condition:
\begin{tabbing}
1.$\quad$ \= $N_0:f(\underline{I}) \Rightarrow_{C_{f_1}} ...\ N_2:f(\underline{X}) \Rightarrow_{C_{f_1}}  ...\   
N_4:f(\underline{X_1}) \Rightarrow_{C_{f_1}}$ \\
2. \> $N_0:f(\underline{I}) \Rightarrow_{C_{f_1}} ...\ N_2:f(\underline{X}) \Rightarrow_{C_{f_1}}  ...\   
N_5:f(\underline{X_2}) \Rightarrow_{C_{f_1}}$ \\
3.\> $N_0:f(\underline{I}) \Rightarrow_{C_{f_1}} ...\ N_2:f(\underline{X}) \Rightarrow_{C_{f_1}}  ...\   
N_6:f(\underline{X_3}) \Rightarrow_{C_{f_1}}$ \\
4.\> $N_0:f(\underline{I}) \Rightarrow_{C_{f_1}} ...\ 
N_4:f(\underline{X_1}) \Rightarrow_{C_{f_2}} N_5:f(\underline{X_2}) \Rightarrow_{C_{f_2}}     
N_6:f(\underline{X_3}) \Rightarrow_{C_{f_2}}$ \\
5.\> $N_0:f(\underline{I}) \Rightarrow_{C_{f_1}} ...\ N_7:f(\underline{X_1}) \Rightarrow_{C_{f_1}}  ...\   
N_9:f(\underline{X_2}) \Rightarrow_{C_{f_1}}$ \\
6.\> $N_0:f(\underline{I}) \Rightarrow_{C_{f_1}} ...\ N_7:f(\underline{X_1}) \Rightarrow_{C_{f_1}}  ...\   
N_{10}:f(\underline{X_3}) \Rightarrow_{C_{f_1}}$ \\
7.\> $N_0:f(\underline{I}) \Rightarrow_{C_{f_1}} ...\ N_2:f(\underline{X}) \Rightarrow_{C_{f_2}}  ...\   
N_9:f(\underline{X_2}) \Rightarrow_{C_{f_2}} N_{10}:f(\underline{X_3}) \Rightarrow_{C_{f_2}}$ \\
8.\> $N_0:f(\underline{I}) \Rightarrow_{C_{f_1}} ...\ N_{11}:f(\underline{X_2}) \Rightarrow_{C_{f_1}}  ...\   
N_{13}:f(\underline{X_3}) \Rightarrow_{C_{f_1}}$ \\
9.\> $N_0:f(\underline{I}) \Rightarrow_{C_{f_1}} ...\ N_2:f(\underline{X}) \Rightarrow_{C_{f_2}}   
N_7:f(\underline{X_1}) \Rightarrow_{C_{f_2}} ...\ N_{13}:f(\underline{X_3}) \Rightarrow_{C_{f_2}}$ \\
10.\> $N_0:f(\underline{I}) \Rightarrow_{C_{f_1}} ...\ N_2:f(\underline{X}) \Rightarrow_{C_{f_2}}   
N_7:f(\underline{X_1}) \Rightarrow_{C_{f_2}} N_{11}:f(\underline{X_2}) \Rightarrow_{C_{f_2}}$ \\
11.\> $N_0:f(\underline{I}) \Rightarrow_{C_{f_2}} N_{14}:f(\underline{X}) \Rightarrow_{C_{f_1}}  ...\   
N_{16}:f(\underline{X_1}) \Rightarrow_{C_{f_1}} ...\ N_{18}:f(\underline{X_2}) \Rightarrow_{C_{f_1}}$ \\
12.\> $N_0:f(\underline{I}) \Rightarrow_{C_{f_2}} N_{14}:f(\underline{X}) \Rightarrow_{C_{f_1}}  ...\   
N_{16}:f(\underline{X_1}) \Rightarrow_{C_{f_1}} ...\ N_{19}:f(\underline{X_3}) \Rightarrow_{C_{f_1}}$ \\
13.\> $N_0:f(\underline{I}) \Rightarrow_{C_{f_2}} ...\   
N_{18}:f(\underline{X_2}) \Rightarrow_{C_{f_2}} N_{19}:f(\underline{X_3}) \Rightarrow_{C_{f_2}}$ \\
14.\> $N_0:f(\underline{I}) \Rightarrow_{C_{f_2}} N_{14}:f(\underline{X}) \Rightarrow_{C_{f_1}}  ...\   
N_{20}:f(\underline{X_2}) \Rightarrow_{C_{f_1}} ...\ N_{22}:f(\underline{X_3}) \Rightarrow_{C_{f_1}}$ \\
15.\> $N_0:f(\underline{I}) \Rightarrow_{C_{f_2}} ...\ N_{16}:f(\underline{X_1}) \Rightarrow_{C_{f_2}}  ...\   
N_{22}:f(\underline{X_3}) \Rightarrow_{C_{f_2}}$ \\
16.\> $N_0:f(\underline{I}) \Rightarrow_{C_{f_2}} ...\ N_{16}:f(\underline{X_1}) \Rightarrow_{C_{f_2}}   
N_{20}:f(\underline{X_2}) \Rightarrow_{C_{f_2}}$ \\
17.\> $N_0:f(\underline{I}) \Rightarrow_{C_{f_2}} ...\ N_{23}:f(\underline{X_1}) \Rightarrow_{C_{f_1}}  ...\   
N_{25}:f(\underline{X_2}) \Rightarrow_{C_{f_1}} ...\ N_{27}:f(\underline{X_3}) \Rightarrow_{C_{f_1}}$ \\
18.\> $N_0:f(\underline{I}) \Rightarrow_{C_{f_2}} N_{14}:f(\underline{X}) \Rightarrow_{C_{f_2}}  ...\   
N_{27}:f(\underline{X_3}) \Rightarrow_{C_{f_2}}$ \\
19.\> $N_0:f(\underline{I}) \Rightarrow_{C_{f_2}} N_{14}:f(\underline{X}) \Rightarrow_{C_{f_2}}  ...\   
N_{25}:f(\underline{X_2}) \Rightarrow_{C_{f_2}}$ \\
20.\> $N_0:f(\underline{I}) \Rightarrow_{C_{f_2}} N_{14}:f(\underline{X}) \Rightarrow_{C_{f_2}}  ...\   
N_{23}:f(\underline{X_1}) \Rightarrow_{C_{f_2}}$ 
\end{tabbing}
Since there is no derivation satisfying 
the conditions of LP-check while violating the term-size decrease condition,
Algorithm \ref{alg1} ends with an answer {\em predicted-terminating} for the query   
$Q_0$. This again is a correct prediction.
\begin{figure*}[htb]
\begin{tabbing}
\input{fig-incomplete2.latex}
\end{tabbing}
\caption{The moded generalized SLDNF-tree of $P_6$ generated by Algorithm \ref{alg1}.}\label{fig-incomplete2} 
\end{figure*}  
\end{example} 

Choosing $r=3$ for Algorithm \ref{alg1}, we are able to obtain correct predictions
for all benchmark programs of the Termination Competition 2007 (see Section \ref{implementation}).
However, we should remark that 
due to the undecidability of the termination problem, 
there exist cases that choosing $r=3$ will lead to an incorrect prediction. 
Consider the following carefully crafted logic program:
\begin{tabbing} 
$\qquad$ \= $P_7:$ $\quad$ \= $p(f(X),Y) \leftarrow  p(X,s(Y)).$ \`$C_{p_1}$\\
\> \>      $p(Z, \underbrace{s(s(...s}_{100 \ items}(0)...))) \leftarrow q.$ \`$C_{p_2}$\\
\>   \>    $q \leftarrow  q.$  \`$C_{q_1}$
\end{tabbing} 
$P_7$ does not terminate for a moded query $Q_0 = p({\cal I}, 0)$, as there is an infinite derivation 
\[N_0:p(\underline{I}, 0) \Rightarrow_{C_{p_1}}  ...\ N_{101}:q \Rightarrow_{C_{q_1}} N_{102}:q \Rightarrow_{C_{q_1}} ...\] 
(see Figure \ref{fig-eg5}) which satisfies conditions (i) and (ii) of Theorem \ref{th-main},
where for any $j \geq 101$, $G_{j+1}$ is a loop goal of $G_j$.
Note that for any repetition number $r$ with $3\leq r\leq 100$, the prefix ending at $N_{r-1}$ 
satisfies both the conditions of LP-check and the term-size decrease property,
where for any $j$ with $0\leq j < r-1$, $G_{j+1}$ is a loop goal of $G_j$. However,
for any $r> 100$, a prefix ending at $N_{100+r}$ 
will be encountered, which satisfies 
the conditions of LP-check but violates the term-size decrease condition,
where for any $j$ with $101\leq j < 100+r$, $G_{j+1}$ is a loop goal of $G_j$.
Therefore, Algorithm \ref{alg1} will return {\em predicted-terminating} for $Q_0$ 
unless $r$ is set above 100.
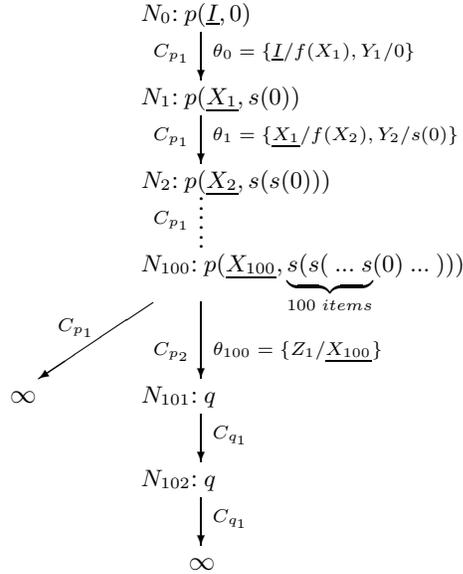
\begin{figure}[htb]
\begin{center}
\setlength{\unitlength}{3947sp}%
\begingroup\makeatletter\ifx\SetFigFont\undefined%
\gdef\SetFigFont#1#2#3#4#5{%
  \reset@font\fontsize{#1}{#2pt}%
  \fontfamily{#3}\fontseries{#4}\fontshape{#5}%
  \selectfont}%
\fi\endgroup%
\begin{picture}(1721,3591)(451,-3406)
\thicklines
{\color[rgb]{0,0,0}\multiput(1651,-1036)(0.00000,-60.00000){6}{\makebox(6.6667,10.0000){\SetFigFont{10}{12}{\rmdefault}{\mddefault}{\updefault}.}}
}%
\thinlines
{\color[rgb]{0,0,0}\put(1651, 14){\vector( 0,-1){300}}
}%
{\color[rgb]{0,0,0}\put(1651,-511){\vector( 0,-1){300}}
}%
{\color[rgb]{0,0,0}\put(1651,-2386){\vector( 0,-1){300}}
}%
{\color[rgb]{0,0,0}\put(1651,-2911){\vector( 0,-1){300}}
}%
{\color[rgb]{0,0,0}\put(1651,-1681){\vector( 0,-1){480}}
}%
{\color[rgb]{0,0,0}\put(1351,-1681){\vector(-3,-2){720}}
}%
\put(1276, 89){\makebox(0,0)[lb]{\smash{{\SetFigFont{9}{10.8}{\rmdefault}{\mddefault}{\updefault}{\color[rgb]{0,0,0}$N_0$:  $p(\underline{I}, 0)$}%
}}}}
\put(1276,-1486){\makebox(0,0)[lb]{\smash{{\SetFigFont{9}{10.8}{\rmdefault}{\mddefault}{\updefault}{\color[rgb]{0,0,0}$N_{100}$:  $p(\underline{X_{100}}, \underbrace{s(s(\ ...\ s}_{100 \ items}(0)\ ...\ )))$ }%
}}}}
\put(1276,-436){\makebox(0,0)[lb]{\smash{{\SetFigFont{9}{10.8}{\rmdefault}{\mddefault}{\updefault}{\color[rgb]{0,0,0}$N_1$:  $p(\underline{X_1},s(0))$ }%
}}}}
\put(1276,-961){\makebox(0,0)[lb]{\smash{{\SetFigFont{9}{10.8}{\rmdefault}{\mddefault}{\updefault}{\color[rgb]{0,0,0}$N_2$:  $p(\underline{X_2},s(s(0)))$ }%
}}}}
\put(1351,-136){\makebox(0,0)[lb]{\smash{{\SetFigFont{7}{8.4}{\rmdefault}{\mddefault}{\updefault}{\color[rgb]{0,0,0}$C_{p_1}$}%
}}}}
\put(1351,-661){\makebox(0,0)[lb]{\smash{{\SetFigFont{7}{8.4}{\rmdefault}{\mddefault}{\updefault}{\color[rgb]{0,0,0}$C_{p_1}$}%
}}}}
\put(1726,-136){\makebox(0,0)[lb]{\smash{{\SetFigFont{7}{8.4}{\rmdefault}{\mddefault}{\updefault}{\color[rgb]{0,0,0}$\theta_0 = \{\underline{I}/f(X_1),Y_1/0\}$}%
}}}}
\put(1726,-661){\makebox(0,0)[lb]{\smash{{\SetFigFont{7}{8.4}{\rmdefault}{\mddefault}{\updefault}{\color[rgb]{0,0,0}$\theta_1 = \{\underline{X_1}/f(X_2),Y_2/s(0)\}$}%
}}}}
\put(1351,-1186){\makebox(0,0)[lb]{\smash{{\SetFigFont{7}{8.4}{\rmdefault}{\mddefault}{\updefault}{\color[rgb]{0,0,0}$C_{p_1}$}%
}}}}
\put(1276,-2836){\makebox(0,0)[lb]{\smash{{\SetFigFont{9}{10.8}{\rmdefault}{\mddefault}{\updefault}{\color[rgb]{0,0,0}$N_{102}$:  $q$ }%
}}}}
\put(1276,-2311){\makebox(0,0)[lb]{\smash{{\SetFigFont{9}{10.8}{\rmdefault}{\mddefault}{\updefault}{\color[rgb]{0,0,0}$N_{101}$:  $q$ }%
}}}}
\put(1351,-2011){\makebox(0,0)[lb]{\smash{{\SetFigFont{7}{8.4}{\rmdefault}{\mddefault}{\updefault}{\color[rgb]{0,0,0}$C_{p_2}$}%
}}}}
\put(1726,-2011){\makebox(0,0)[lb]{\smash{{\SetFigFont{7}{8.4}{\rmdefault}{\mddefault}{\updefault}{\color[rgb]{0,0,0}$\theta_{100} = \{Z_1/\underline{X_{100}}\}$}%
}}}}
\put(1726,-2536){\makebox(0,0)[lb]{\smash{{\SetFigFont{7}{8.4}{\rmdefault}{\mddefault}{\updefault}{\color[rgb]{0,0,0}$C_{q_1}$}%
}}}}
\put(1726,-3061){\makebox(0,0)[lb]{\smash{{\SetFigFont{7}{8.4}{\rmdefault}{\mddefault}{\updefault}{\color[rgb]{0,0,0}$C_{q_1}$}%
}}}}
\put(1576,-3361){\makebox(0,0)[lb]{\smash{{\SetFigFont{10}{12.0}{\rmdefault}{\mddefault}{\updefault}{\color[rgb]{0,0,0}$\infty$ }%
}}}}
\put(751,-1861){\makebox(0,0)[lb]{\smash{{\SetFigFont{7}{8.4}{\rmdefault}{\mddefault}{\updefault}{\color[rgb]{0,0,0}$C_{p_1}$}%
}}}}
\put(451,-2311){\makebox(0,0)[lb]{\smash{{\SetFigFont{10}{12.0}{\rmdefault}{\mddefault}{\updefault}{\color[rgb]{0,0,0}$\infty$ }%
}}}}
\end{picture}%
\end{center}
\caption{The moded generalized SLDNF-tree of $P_7$ 
with a moded query $p({\cal I}, 0)$.}\label{fig-eg5} 
\end{figure}

The following result shows that choosing a sufficiently large 
repetition number guarantees the correct prediction 
for non-terminating programs.

\begin{theorem}
\label{prec-alg}
Let $P$ be a logic program and $Q$ be a query such that $P$ is non-terminating for $Q$.
There always exists a number $R$ such that Algorithm \ref{alg1} with any repetition number $r\geq R$ 
produces the answer {\em predicted-non-terminating}.
\end{theorem}

\begin{proof}
Let us assume the contrary. That is, we assume 
that for any number $N$, there exists a larger number $r$
such that Algorithm \ref{alg1} for $P$ with query $Q$ and 
repetition number $r$ produces the answer {\em predicted-terminating}.
This means that for all $r\geq 2$
the prefix of form (\ref{eq3}) of each infinite branch $D$
in the moded generalized SLDNF-tree $MT_Q$  
satisfies the term-size decrease property. According 
to Theorem \ref{th-term-dec}, $D$ has an infinite chain of 
substitutions of form (\ref{ins-chain})
for some input variable $I$ at $Q$. 
This means that $D$ does not satisfy condition (ii) of Theorem \ref{th-main}. 
However, since $P$ is non-terminating for $Q$, 
by Corollary \ref{cor-main} $MT_Q$  
has at least one infinite branch of form (\ref{eq2}) 
satisfying conditions (i) and (ii) of Theorem \ref{th-main}. 
We then have a contradiction
and thus conclude the proof. 
\end{proof}

The same result applies for any concrete query $Q$. That is,
there always exists a number $R$ such that 
Algorithm \ref{alg1} with any $r\geq R$ 
produces the answer {\em terminating} or 
{\em predicted-terminating} when $P$ is terminating for $Q$. 
The proof for this is simple.
When $P$ is terminating for a concrete query $Q$, the (moded) generalized SLDNF-tree for $Q$ is finite.
Let $R$ be the number of nodes of the longest branch in the tree. For any $r\geq R$,
Algorithm \ref{alg1} will produce the answer 
{\em terminating} or {\em predicted-terminating},
since no branch will be cut by LP-check. 

However, whether the above claim holds for any moded query $Q$
when $P$ is terminating for $Q$ remains an interesting 
open problem.

\section{Experimental Results}
\label{implementation}
We have evaluated our termination prediction technique on a benchmark of 301 Prolog programs. 
In this section, we first describe the benchmark and our experimental results using a straightforward 
implementation of Algorithm \ref{alg1}. Then, we define a pruning 
technique to reduce the size of moded generalized SLDNF-trees generated for our prediction.
Finally, we make a comparison between the state-of-the-art 
termination and non-termination analyzers and our termination prediction tool.

Our benchmark consists of 301 programs with moded queries from the Termination 
Competition 2007 
(http://www.lri.fr/\verb+~+marche/termination-competition/2007). Only 23 programs of the competition
are omitted because they contain non-logical operations such as arithmetics (for 
most of these programs neither termination nor non-termination 
could be shown by any of the tools in the competition). The 
benchmark contains 244 terminating programs and 57 non-terminating ones.
The most accurate termination analyzer of the competition, AProVE \cite{GST06}, 
proves termination of 238 benchmark programs. The non-termination 
analyzer NTI \cite{PM06,Pay06} proves non-termination of 42 programs. Because the prediction 
does not produce a termination or non-termination proof,
our goal is to outperform the analyzers of the competition
by providing a higher number of correct predictions.

We implemented our tool, TPoLP: Termination Prediction of Logic Programs,
in SWI-prolog (http://www.swi-prolog.org). TPoLP is freely available from 
http://www.cs.kuleuven .be/\verb!~!dean. The moded 
generalized SLDNF-tree is generated following 
Algorithm \ref{alg1}. 
It is initialized with the moded query and extended
until all branches are cut or a timeout occurs. To improve the 
efficiency of the analysis, a number of optimalizations were implemented, 
such as constant time access to the nodes and the arcs of the derivations. 
The experiments have been performed 
using SWI-Prolog 5.6.40 (http://www.swi-prolog.org), on an Intel Core2 Duo 2,33GHz, 2 Gb RAM.

Table \ref{table:nopruning} gives an overview of the predictions with repetition
numbers $r=2$, $r=3$, and $r=4$. As we mentioned earlier,
$r=2$ does not suffice because 
some of the predictions are wrong and we want high reliable predictions. When 
$r$ is set above two, all predictions made for the benchmark are correct.
This shows that in practice, there is no need to increase the repetition number any further.
\begin{table}[htb]
\begin{tabular}{|r||c|c|c|}  \cline{1-4}
               &	$\quad r = 2	\quad$ & 	$\quad r = 3	\quad$ &$\quad r = 4\quad$ \\ \cline{1-4}
\cline{1-4}
Correct predictions	&	 291 	&	 271 	&	 234  \\
Wrong predictions	&	 7 		&	 0 		&	 0  \\
Out of time/memory	&	 3 		&	 30 	&	 67  \\
Average time (Sec)	&	 1.7 	&	 24.9 	&	 59.3  \\ 
\cline{1-4}
\end{tabular}
\caption{{\em Prediction with different repetition numbers.}}
\label{table:nopruning}
\end{table}

When we increase the repetition number, the cost of prediction increases as well.
Table \ref{table:nopruning} shows that for $r=3$,
about 10\% of the programs break the time limit of four minutes, 
and for $r=4$, about 20\% break the limit.

The component of the algorithm taking most of the time differs from program to program.
When a lot of branches are cut by LP-check, constructing the LP cuts is usually
the bottleneck. For programs with a low amount of LP cuts, most of the time is spent 
on constructing the SLDNF-derivations. Some of the derivations count more than a million 
nodes. To overcome such performance issues, we implemented the following 
pruning technique on loop goals.

\begin{definition} [Pruning variants]
\label{prune-1}
Let $G_2$ be a loop goal of $G_1$ for which the selected subgoals are
variants. Then, all clauses that 
have already been applied at $G_2$ are skipped at $G_1$ during backtracking.
\end{definition}

The idea of this pruning is simple. For loop goals with variant
selected subgoals, applying the non-looping clauses to them will generate
the same derivations below them with the same termination properties.
Therefore, the derivations already generated below $G_2$
need not be regenerated at $G_1$ during backtracking.

For the sake of efficiency, in our implementation we determine variants 
by checking that they have the same symbol string.

Consider Example \ref{incomplete2} again. When the above pruning mechanism
is applied, Algorithm \ref{alg1} will simplify the moded generalized SLDNF-tree
of Figure \ref{fig-incomplete2} into Figure \ref{figure:pruning}.
The pruning takes place at $N_2$ and $N_0$, where $G_4$ is a loop goal of $G_2$ 
that is a loop goal of $G_0$ and their selected subgoals are variants.
\begin{figure*}[ht]
\setlength{\unitlength}{2723sp}%
\begingroup\makeatletter\ifx\SetFigFont\undefined%
\gdef\SetFigFont#1#2#3#4#5{%
  \reset@font\fontsize{#1}{#2pt}%
  \fontfamily{#3}\fontseries{#4}\fontshape{#5}%
  \selectfont}%
\fi\endgroup%
\begin{picture}(6838,3190)(901,-3017)
\thinlines
{\color[rgb]{0,0,0}\put(2021,-1111){\vector(-2,-1){336}}
}%
{\color[rgb]{0,0,0}\put(2431,-1111){\vector( 4,-1){1037.647}}
}%
{\color[rgb]{0,0,0}\put(1276,-1516){\vector( 0,-1){270}}
}%
{\color[rgb]{0,0,0}\put(1276,-2011){\vector( 0,-1){270}}
}%
{\color[rgb]{0,0,0}\put(1276,-2536){\vector( 0,-1){270}}
}%
{\color[rgb]{0,0,0}\put(5326, 14){\vector( 4,-1){2078.823}}
}%
{\color[rgb]{0,0,0}\put(4801, 14){\vector(-4,-1){1994.118}}
}%
{\color[rgb]{0,0,0}\put(2596,-736){\vector(-2,-1){336}}
}%
\put(6826,-511){\makebox(0,0)[lb]{\smash{{\SetFigFont{14}{16.8}{\rmdefault}{\mddefault}{\updefault}{\color[rgb]{0,0,0}X}%
}}}}
\put(3076,-1411){\makebox(0,0)[lb]{\smash{{\SetFigFont{14}{16.8}{\rmdefault}{\mddefault}{\updefault}{\color[rgb]{0,0,0}X}%
}}}}
\put(1576,-1111){\makebox(0,0)[lb]{\smash{{\SetFigFont{6}{7.2}{\rmdefault}{\mddefault}{\updefault}{\color[rgb]{0,0,0}$C_{f_1}$}%
}}}}
\put(901,-1411){\makebox(0,0)[lb]{\smash{{\SetFigFont{6}{7.2}{\rmdefault}{\mddefault}{\updefault}{\color[rgb]{0,0,0}$N_3$:  $g(s(s(s(\underline{X}))))$}%
}}}}
\put(2326,-1261){\makebox(0,0)[lb]{\smash{{\SetFigFont{6}{7.2}{\rmdefault}{\mddefault}{\updefault}{\color[rgb]{0,0,0}$C_{f_2}$}%
}}}}
\put(1876,-1036){\makebox(0,0)[lb]{\smash{{\SetFigFont{6}{7.2}{\rmdefault}{\mddefault}{\updefault}{\color[rgb]{0,0,0}$N_2$:  $f(\underline{X})$}%
}}}}
\put(1351,-1636){\makebox(0,0)[lb]{\smash{{\SetFigFont{6}{7.2}{\rmdefault}{\mddefault}{\updefault}{\color[rgb]{0,0,0}$\underline{X}/s(X_1)$}%
}}}}
\put(976,-1636){\makebox(0,0)[lb]{\smash{{\SetFigFont{6}{7.2}{\rmdefault}{\mddefault}{\updefault}{\color[rgb]{0,0,0}$C_{g_1}$}%
}}}}
\put(901,-2986){\makebox(0,0)[lb]{\smash{{\SetFigFont{6}{7.2}{\rmdefault}{\mddefault}{\updefault}{\color[rgb]{0,0,0}$N_6$:  $f(\underline{X_3})$}%
}}}}
\put(901,-2461){\makebox(0,0)[lb]{\smash{{\SetFigFont{6}{7.2}{\rmdefault}{\mddefault}{\updefault}{\color[rgb]{0,0,0}$N_5$:  $f(\underline{X_2})$}%
}}}}
\put(1351,-2686){\makebox(0,0)[lb]{\smash{{\SetFigFont{6}{7.2}{\rmdefault}{\mddefault}{\updefault}{\color[rgb]{0,0,0}$\underline{X_2}/s(X_3)$}%
}}}}
\put(976,-2686){\makebox(0,0)[lb]{\smash{{\SetFigFont{6}{7.2}{\rmdefault}{\mddefault}{\updefault}{\color[rgb]{0,0,0}$C_{f_2}$}%
}}}}
\put(976,-2161){\makebox(0,0)[lb]{\smash{{\SetFigFont{6}{7.2}{\rmdefault}{\mddefault}{\updefault}{\color[rgb]{0,0,0}$C_{f_2}$}%
}}}}
\put(1351,-2161){\makebox(0,0)[lb]{\smash{{\SetFigFont{6}{7.2}{\rmdefault}{\mddefault}{\updefault}{\color[rgb]{0,0,0}$\underline{X_1}/s(X_2)$}%
}}}}
\put(901,-1936){\makebox(0,0)[lb]{\smash{{\SetFigFont{6}{7.2}{\rmdefault}{\mddefault}{\updefault}{\color[rgb]{0,0,0}$N_4$:  $f(\underline{X_1})$}%
}}}}
\put(3376,-211){\makebox(0,0)[lb]{\smash{{\SetFigFont{6}{7.2}{\rmdefault}{\mddefault}{\updefault}{\color[rgb]{0,0,0}$C_{f_1}$}%
}}}}
\put(4726, 89){\makebox(0,0)[lb]{\smash{{\SetFigFont{6}{7.2}{\rmdefault}{\mddefault}{\updefault}{\color[rgb]{0,0,0}$N_0$:  $f(\underline{I})$}%
}}}}
\put(2026,-661){\makebox(0,0)[lb]{\smash{{\SetFigFont{6}{7.2}{\rmdefault}{\mddefault}{\updefault}{\color[rgb]{0,0,0}$N_1$:  $g(s(s(s(\underline{I}))))$}%
}}}}
\put(2551,-886){\makebox(0,0)[lb]{\smash{{\SetFigFont{6}{7.2}{\rmdefault}{\mddefault}{\updefault}{\color[rgb]{0,0,0}$\underline{I}/s(X)$}%
}}}}
\put(6001,-361){\makebox(0,0)[lb]{\smash{{\SetFigFont{6}{7.2}{\rmdefault}{\mddefault}{\updefault}{\color[rgb]{0,0,0}$C_{f_2}$}%
}}}}
\put(2026,-811){\makebox(0,0)[lb]{\smash{{\SetFigFont{6}{7.2}{\rmdefault}{\mddefault}{\updefault}{\color[rgb]{0,0,0}$C_{g_1}$}%
}}}}
\put(7201,-736){\makebox(0,0)[lb]{\smash{{\SetFigFont{8}{9.6}{\rmdefault}{\mddefault}{\updefault}{\color[rgb]{0,0,0}Pruned}%
}}}}
\put(3376,-1561){\makebox(0,0)[lb]{\smash{{\SetFigFont{8}{9.6}{\rmdefault}{\mddefault}{\updefault}{\color[rgb]{0,0,0}Pruned}%
}}}}
\end{picture}%
\caption{Figure \ref{fig-incomplete2} is simplified with pruning.}
\label{figure:pruning}
\end{figure*}
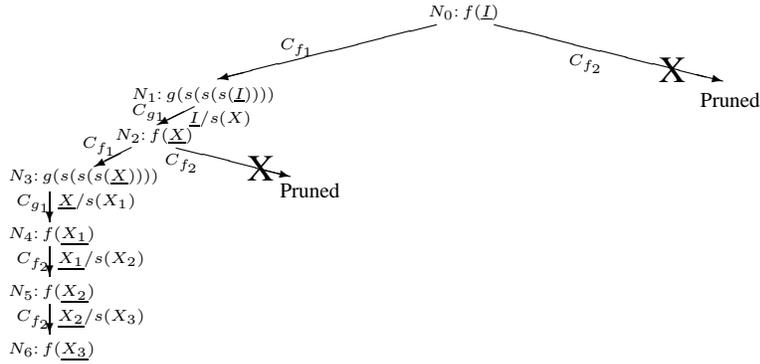

A stronger version of the above pruning mechanism can be obtained 
by removing the condition in Definition \ref{prune-1}:
for which the selected subgoals are variants. That is,
we do not require the selected subgoals of loop goals to be variants.
We call this version {\em Pruning loop goals}.

Table \ref{table:pruning} gives an overview of our predictions with $r=3$ 
as the repetition number in the cases of no pruning, pruning variants, 
and pruning loop goals. The table shows that pruning is a good 
tradeoff between the accuracy and the efficiency of the prediction.
When applying the variants pruning mechanism the size of the derivations drops
considerably, while all predictions for the benchmark are still correct.
Due to the pruning, more than 98\% of the predictions
finish within the time limit.
Applying the loop goals pruning mechanism leads to a greater reduction
in the size of derivations. However,
in this case we sacrifice accuracy: three non-terminating programs
are predicted to be terminating. 
\begin{table}[htb]
 \begin{tabular}{|r||c|c|c|}
\cline{1-4}
			&	No pruning	&	Pruning variants	& 	Pruning loop goals	\\
\cline{1-4}
Correct predictions 	&	271		&		296		&		297		\\
Wrong predictions	&	0		&		0		&		3		\\
Out of time/memory	&	30		&		5		&		1		\\	
Average time (Sec)	&	24.9		&		4.4		&		0.05		\\	
\cline{1-4}
\end{tabular}
\caption{{\em The effect of pruning.}}
\label{table:pruning}
\end{table}

Table \ref{table:comparison} gives a comparison between our predictions
(with $r=3$ and the variants pruning mechanism) and the proving results of the
state-of-the-art termination and non-termination analyzers.
Note that our tool, TPoLP, is the only
tool which analyzes both for 
termination and non-termination of logic programs. The results are very encouraging.
We correctly predict the termination property of all 
benchmark programs except for five programs which broke the time limit.
It is also worth noticing that for all programs of the benchmark,
either an existing analyzer finds a termination or non-termination proof,
or a correct prediction is made by our tool. This shows that
our prediction tool can be a very useful addition to any termination
or non-termination analyzer.

\begin{table}[htb]
 \begin{tabular}{|r||c|c|c|c|c|}
\cline{1-6}
	& \multirow{2}{*}{TPoLP prediction} &	\multicolumn{4}{c|} {$\ $Termination/non-termination proof $\ $}	\\ \cline{3-6}
       &   &     AProVE 	& NTI 	& Polytool  & TALP  \\
\cline{1-6}
Answer {\em Terminating} (244)  	&	239		&	238	&	0	&	206		&	164	\\
Answer {\em Non-terminating}	(57) &	57		&	0	&	42	&	0		&	0	\\
\cline{1-6}
\end{tabular}
\caption{{\em Comparison between TPoLP and the existing analyzers.}}
\label{table:comparison}
\end{table}

\section{Related Work}
\label{sec5}
Most existing approaches to the termination problem are {\em norm-} or {\em level 
mapping-based} in the sense that they perform termination analysis by building from 
the source code of a logic program some well-founded 
termination conditions/constraints in terms of 
norms (i.e. term sizes of atoms of clauses), level mappings,
interargument size relations and/or instantiation dependencies, which when solved, 
yield a termination proof (see, e.g., \citeN{DD93} for a survey and more recent papers 
\cite{Apt93,BCRE02,BCGGV05,DDV99,GC05,LS97,Mar96-1,MN01}). 
Another main stream is {\em transformational} approaches,
which transform a logic program into a term rewriting
system (TRS) and then analyze the termination property of
the resulting TRS instead 
\cite{AM93,AZ96,GST06,Mar96,OCM00,RKS98,SGST06,vanR97}. 
All of these approaches are used for a termination proof;
i.e., they compute sufficient termination conditions which once 
satisfied, lead to a positive conclusion {\em terminating}.
Recently, \citeN{PM06} and \citeN{Pay06} propose an approach to 
computing sufficient non-termination conditions
which when satisfied, lead to a negative conclusion {\em non-terminating}. 
A majority of these termination/non-termination proof approaches
apply only to positive logic programs.

Our approach presented in this paper differs significantly
from existing termination analysis approaches.
First, we do not make a termination proof, nor do we 
make a non-termination proof. Instead, we make a termination 
prediction (see Figure \ref{framework}) $-$ a heuristic approach to attacking the 
undecidable termination problem.
Second, we do not rely on static norms or level 
mappings, nor do we transform a logic program to a term rewriting system.
Instead, we focus on detecting infinite SLDNF-derivations
with the understanding that a logic program is terminating
for a query if and only if there is no infinite SLDNF-derivation
with the query. We have established a necessary and sufficient
characterization of infinite
SLDNF-derivations with arbitrary (concrete or moded) queries,
introduced a new loop checking mechanism,
and developed an algorithm that predicts termination of 
general logic programs with arbitrary queries by identifying
potential infinite SLDNF-derivations.
Since the algorithm implements the necessary and sufficient
conditions (the characterization) of
an infinite SLDNF-derivation, its prediction is very effective.
Our experimental results show that except for five programs
which break the time limit, our prediction is $100\%$
correct for all 296 benchmark programs of the Termination Competition 2007,
of which eighteen programs cannot be proved by any of the existing
state-of-the-art analyzers like 
AProVE07 \cite{GST06}, NTI \cite{PM06,Pay06}, Polytool \cite{ND05,NDGS06} and TALP \cite{OCM00}. 
  
Our termination prediction approach uses a loop checking mechanism (a loop check)
to implement a characterization of infinite SLDNF-derivations. 
Well-known loop checks include VA-check \cite{BAK91,VG87}, 
OS-check \cite{BDM92,MD96,Sa93}, and VAF-checks \cite{shen97,shen001}.
All apply to positive logic programs. In particular,
VA-check applies to function-free logic programs, where an 
infinite derivation is characterized by a sequence of selected {\em variant subgoals}. 
OS-check identifies an infinite derivation with a sequence of selected subgoals 
with the same predicate symbol {\em whose sizes do not decrease}.
VAF-checks take a sequence of selected {\em expanded variant subgoals} 
as major characteristics of an infinite derivation. Expanded variant subgoals are variant
subgoals except that some terms may grow bigger. In this paper, 
a new loop check mechanism, LP-check, is introduced in which 
an infinite derivation is identified with a sequence of {\em loop goals}.
Most importantly, enhancing
LP-check with the term-size decrease property
leads to the first loop check for moded queries.

\section{Conclusion and Future Work}
\label{sec6}
We have presented a heuristic framework
for attacking the undecidable termination problem of
logic programs, as an alternative to current termination/non-termination
proof approaches. We introduced an idea of termination prediction,
established a necessary and sufficient
characterization of infinite
SLDNF-derivations with arbitrary (concrete or moded) queries,
built a new loop checking mechanism,
and developed an algorithm that predicts termination of 
general logic programs with arbitrary queries.
We have implemented a termination prediction tool, TPoLP, and obtained 
quite satisfactory experimental results. Except for five programs
which break the experiment time limit, our prediction is $100\%$
correct for all 296 benchmark programs of the Termination Competition 2007. 

Our prediction approach can be used standalone, e.g., it may
be incorporated into Prolog as a termination debugging tool;
or it is used along with some termination/non-termination proof tools
(see the framework in Figure \ref{framework}). 

Limitations of the current prediction approach include that it cannot handle
floundering queries and programs with non-logical operators.
To avoid floundering, 
we assume that no negative subgoals containing either ordinary or input variables are
selected at any node of a moded generalized SLDNF-tree
(violation of the assumption can easily be checked
in the course of constructing generalized SLDNF-trees).
This assumption seems able to be relaxed by allowing input variables 
to occur in selected negative subgoals. This makes us able to predict termination of
programs like
\begin{tabbing} 
$\qquad$ \= $P:$ $\quad$ \= $p(X)\leftarrow \neg q(X)$. \\
\> \> $q(a)\leftarrow q(a)$.     
\end{tabbing} 
which is non-terminating for the moded query $p({\cal I})$.

Our future work includes further improvement of the prediction efficiency
of TPoLP. As shown in Table \ref{table:pruning},
there are five benchmark programs breaking our experiment time limit. 
We are also considering extensions of the proposed 
termination prediction to typed queries \cite{BCGGV05}
and to logic programs with tabling \cite{chen96,shen004,VDK001}.

\section{Acknowledgments}
We would like to thank the anonymous referees for their constructive comments and 
suggestions that helped us improve this work. 
Yi-Dong Shen is supported in part by 
NSFC grants 60673103,
60721061 and 60833001, and by the National High-tech R\&D Program (863 Program).
Dean Voets is supported by the Flemish Fund for Scientific Research
- FWO-project G0561-08.

\bibliographystyle{acmtrans}
\bibliography{tplp4}

\end{document}